\documentclass[11pt,times,authoryear]{elsarticle}
\usepackage{fullpage}
\usepackage{amssymb}
\usepackage{amsthm}
\usepackage[T1]{fontenc}
\usepackage{hyperref}
\usepackage[]{geometry}
\usepackage{amsmath}
\usepackage[latin1]{inputenc}
\usepackage{xcolor,xspace}
\usepackage[english]{babel}
\usepackage{mathptmx}
\usepackage{siunitx}

\makeatletter
\newtheorem*{rep@theorem}{\rep@title}
\newcommand{\newreptheorem}[2]{%
\newenvironment{rep#1}[1]{%
 \def\rep@title{#2 \ref{##1} (restated)}%
 \begin{rep@theorem}}%
 {\end{rep@theorem}}}
\makeatother
\makeatletter
\newtheorem*{prov@theorem}{\prov@title}
\newcommand{\newprovtheorem}[2]{%
\newenvironment{prov#1}[1]{%
 \def\prov@title{#2 \ref{##1} (proven)}%
 \begin{prov@theorem}}%
 {\end{prov@theorem}}}
\makeatother

\newtheorem{theorem}{Theorem}
\newreptheorem{theorem}{Theorem}
\newprovtheorem{theorem}{Theorem}
\newtheorem{lemma}[theorem]{Lemma}
\newtheorem{proposition}[theorem]{Proposition}
\newreptheorem{proposition}{Proposition}
\newprovtheorem{proposition}{Proposition}
\newtheorem{corollary}[theorem]{Corollary}
\theoremstyle{definition}
\newtheorem{remark}[theorem]{Remark}
\newtheorem{definition}{Definition}
\newtheorem{algo}{Algorithm}
\newtheorem{assumption}{Assumption}

\newcommand\IP{\texttt{IP}\xspace}
\newcommand\IPS{\texttt{IP1S}\xspace}

\newcommand\MI{$D$-\texttt{OSMC}\xspace}
\newcommand\SMC{\texttt{SMC}\xspace}
\newcommand\PE{\texttt{PolyEquiv}\xspace}
\newcommand\PP{\texttt{PolyProj}\xspace}
\newcommand\MPKC{\texttt{MPKC}\xspace}

\newcommand\NoSol{\textsc{NoSolution}\xspace}
\newcommand\GPD{\textsc{GPD}\xspace}

\newcommand\mcal{\mathcal}
\newcommand\mbf{\mathbf}
\newcommand\mbb{\mathbb}
\newcommand\msf{\mathsf}
\newcommand\mrm{\mathrm}
\newcommand\mfk{\mathfrak}
\newcommand\POW{\mrm{\mbf{POW}}}
\newcommand\eps{\varepsilon}
\newcommand\f{\mbf{f}}
\newcommand\g{\mbf{g}}
\newcommand\x{\mbf{x}}

\newcommand\Z{\mbb{Z}}
\newcommand\K{\mbb{K}}
\newcommand\LL{\mbb{L}}

\newcommand\F{\mbb{F}}
\newcommand\Q{\mbb{Q}}

\newcommand\cH{\mcal{H}}
\newcommand\cK{\mcal{K}}

\newcommand\cC{\mcal{C}}

\newcommand\id{\mrm{Id}}
\newcommand\rJ{\mrm{J}}
\newcommand\vphi{\varphi}
\newcommand\ceil[1]{\left\lceil#1\right\rceil}
\newcommand\floor[1]{\left\lfloor#1\right\rfloor}

\newcommand\pare[1]{\left(#1\right)}
\newcommand\acc[1]{\left\{#1\right\}}

\newcommand\Gauss{Gau\ss}

\DeclareMathOperator\car{char}

\DeclareMathOperator\GL{GL}
\DeclareMathOperator\T{T}
\DeclareMathOperator\Cent{\mcal{C}}

\DeclareMathOperator\Diag{Diag}

\DeclareMathOperator\tr{Tr}
\DeclareMathOperator\rank{rank}

\journal{Journal of Complexity}

\begin{document}

\begin{frontmatter}


%\tnotetext[HPAC]{}
\title{Polynomial-Time Algorithms for Quadratic Isomorphism of Polynomials:
  The Regular Case}%\tnoteref{HPAC}}

\author[label1,label2,label3]{J\'er\'emy
  Berthomieu\corref{cor1}}
\cortext[cor1]{Laboratoire d'Informatique de Paris 6,
  Université Pierre-et-Marie-Curie, Boîte Courrier 169, 4 place
  Jussieu, F-75252 Paris Cedex 05, France.}
\ead{jeremy.berthomieu@lip6.fr}
\author[label3,label1,label2]{Jean-Charles Faug\`ere}
\ead{jean-charles.faugere@inria.fr}
\author[label1,label2,label3]{Ludovic Perret}
\ead{ludovic.perret@lip6.fr}
\address[label1]{Sorbonne Universit\'es, \textsc{UPMC} Univ Paris 06,
  \'Equipe \textsc{PolSys}, \textsc{LIP6}, F-75005, Paris, France}
\address[label2]{\textsc{CNRS}, \textsc{UMR 7606}, \textsc{LIP6},
  F-75005, Paris, France}
\address[label3]{\textsc{INRIA}, \'Equipe \textsc{PolSys},
  Centre Paris~-- Rocquencourt, F-75005, Paris, France}

% \author{}

% \address{}

\begin{abstract}
  Let $\mathbf{f}=(f_1,\ldots,f_m)$ and $\mathbf{g}=(g_1,\ldots,g_m)$ be
  two sets of $m\geq 1$ nonlinear polynomials in
  $\mathbb{K}[x_1,\ldots,x_n]$ (\(\mathbb{K}\) being a field). We
  consider the computational
  problem of finding -- if any -- an invertible transformation on
  the variables mapping $\mathbf{f}$ to $\mathbf{g}$. The
  corresponding equivalence problem is known as \emph{Isomorphism of
    Polynomials with one Secret} (\texttt{IP1S}) and is a fundamental
  problem in multivariate cryptography. Amongst its applications, we
  can cite Graph Isomorphism (\texttt{GI}) which reduces to
  equivalence of cubic polynomials with respect to an invertible
  linear change of variables, according to Agrawal and Saxena.
  The main result is a randomized polynomial-time
  algorithm for solving \texttt{IP1S} for quadratic instances~-- a
  particular case of importance in cryptography.
  
  To this end, we show that \texttt{IP1S} for quadratic polynomials
  can be reduced to a variant of 
  the classical module isomorphism problem in representation theory.
  We show that we can essentially \emph{linearize} the problem by
  reducing  quadratic-\texttt{IP1S} to test
  the orthogonal simultaneous similarity of symmetric matrices; this
  latter problem was shown by Chistov, Ivanyos and Karpinski
  (\textsc{ISSAC} 1997) to be
  equivalent to finding an invertible matrix in the linear space
  $\mathbb{K}^{n \times n}$ of $n \times n$ matrices over
  $\mathbb{K}$ and 
  to compute the square root in a certain representation
  in a matrix algebra. While computing
  square roots of matrices can be done efficiently using numerical
  methods, it seems difficult to control the bit complexity of such
  methods. However, we present exact and polynomial-time algorithms
  for computing a representation of the square root of a matrix in
  $\mathbb{K}^{n \times n}$,  for
  various fields (including finite fields), as a product of two
  matrices. Each
  coefficient of these matrices lie in an extension field of
  $\mathbb{K}$ of polynomial degree. We then consider
  \#\texttt{IP1S}, the counting version of \texttt{IP1S} for quadratic
  instances. In particular, we provide a (complete) characterization
  of the automorphism group of homogeneous quadratic polynomials.
  Finally, we also consider the more general \emph{Isomorphism of
    Polynomials} (\texttt{IP}) problem where we allow an invertible
  linear transformation on the variables \emph{and} on the set of
  polynomials. A randomized polynomial-time algorithm for solving
  \texttt{IP} when \(\mathbf{f}=(x_1^d,\ldots,x_n^d)\) is
  presented. From an algorithmic point of view, the problem boils
  down to factoring the determinant of a linear matrix (\emph{i.e.}\
  a matrix whose components are linear polynomials).  This extends
  to \texttt{IP} a result of Kayal obtained for \texttt{PolyProj}.
\end{abstract}


\end{frontmatter}


\section{Introduction}
A fundamental question in computer science is to provide algorithms
allowing to test if two given objects are \emph{equivalent} with
respect to some transformation. In this paper, we consider equivalence
of nonlinear polynomials in several variables.  Equivalence of
polynomials has profound connections with a rich variety of
fundamental problems in computer science, ranging -- among others
topics -- from cryptography
(\emph{e.g.}~\cite{DBLP:conf/eurocrypt/Patarin96,
  DBLP:conf/nss/TangX12,Tang2014,DBLP:conf/ispec/YangTY11}),
arithmetic complexity (\emph{via} Geometric Complexity Theory (GCT)
for instance,
see~\cite{DBLP:conf/coco/Burgisser12,DBLP:conf/stoc/Kayal12,
  DBLP:journals/cacm/Mulmuley12,DBLP:journals/siamcomp/MulmuleyS01}),
testing low degree affine-invariant properties~
(\cite{DBLP:conf/soda/BhattacharyyaFL13,
  DBLP:journals/cdm/GreenT09,DBLP:journals/eccc/BhattacharyyaGRS11},
$\ldots$).  As we will see, the notion of equivalence can come with
different flavours that impact the intrinsic hardness of the problem
considered.

In~\cite{DBLP:conf/stacs/AgrawalS06,Saxena2006}, the authors show that
Graph Isomorphism reduces to equivalence of cubic polynomials with
respect to an invertible linear change of variables (a similar
reduction holds between $\F$-algebra Isomorphism and cubic equivalence
of polynomials). This strongly suggests that solving equivalence
problems efficiently is a very challenging algorithmic task.

In cryptography, the hardness of deciding equivalence between two sets
of $m$ polynomials with respect to an invertible linear change of
variables is the security core of several cryptographic schemes: the
seminal zero-knowledge \textsc{ID} scheme
of~\cite{DBLP:conf/eurocrypt/Patarin96}, and more
recently group/proxy signature
schemes~(\cite{DBLP:conf/nss/TangX12,Tang2014,DBLP:conf/ispec/YangTY11}). Note
that there is a subtle difference between the equivalence problem
considered in~\cite{DBLP:conf/stacs/AgrawalS06,Kayal2011,Saxena2006} and
the one considered in cryptographic applications.

Whilst~\cite{DBLP:conf/stacs/AgrawalS06,Kayal2011,Saxena2006} restrict
their attention to $m=1$, arbitrary $m \geq 1$ is usually considered
in cryptographic applications. In the former case, the problem is
called \emph{Polynomial Equivalence} (\PE), whereas it is called
\emph{Isomorphism of Polynomials with One Secret} (\IPS) problem in
the latter case. We emphasize that the hardness of equivalence can
drastically vary in function of $m$. An interesting example is the
case of quadratic forms. The problem is completely solved when $m=1$,
but no polynomial-time algorithm exists for deciding simultaneous
equivalence of quadratic forms. In this paper, we make a step ahead to
close this gap by presenting a randomized polynomial-time algorithm
for solving simultaneous equivalence of quadratic forms over various
fields.

Equivalence of multivariate polynomials is also a fundamental problem
in Multivariate Public-Key Cryptography (\MPKC). This is a family of
asymmetric (encryption and signature) schemes whose public-key is
given by a set of $m$ multivariate
equations~(\cite{C*,DBLP:conf/eurocrypt/Patarin96}). To minimize the
public-key storage, the multivariate polynomials considered are
usually quadratic.  The basic idea of \MPKC is to construct a
public-key which is equivalent to a set of quadratic multivariate
polynomials with a specific structure (see for instance~\cite{WP11}).
Note that the notion of equivalence considered in this context is more
general than the one considered for \PE or \IPS. Indeed, the
equivalence is induced by an invertible linear change of variables and
an invertible linear combination on the polynomials. The corresponding
equivalence problem is
known~(\cite{DBLP:conf/eurocrypt/Patarin96}) as
\emph{Isomorphism of Polynomials} (\IP or \texttt{IP2S}).
      
\PE, \IP, and \IPS are not NP-Hard unless the
polynomial-hierarchy
collapses, \cite{DBLP:journals/eccc/ECCC-TR04-116,DBLP:conf/eurocrypt/PatarinGC98}. However,
the situation changes drastically when considering the equivalence for
more general linear transformations (in particular, not necessarily
invertible). In this context, the problem is called \PP. At {\sc
  SODA'11}, \cite{Kayal2011} showed that \PP is
NP-Hard. This may be due to the fact that various fundamental
questions in arithmetic complexity can be re-interpreted as particular
instances of \PP (see~\cite{DBLP:conf/coco/Burgisser12,
  DBLP:conf/stoc/Kayal12,DBLP:journals/cacm/Mulmuley12,
  DBLP:journals/siamcomp/MulmuleyS01}).

Typically, the famous \textsc{VP} vs \textsc{VNP}
question~(\cite{DBLP:journals/tcs/Valiant79}) can be formulated as an
equivalence problem between the determinant and permanent
polynomials. Such a link is in fact the core motivation of Geometric
Complexity Theory. The problem of computing the symmetric
rank~(\cite{DBLP:journals/jsc/BernardiGI11,DBLP:journals/siammax/ComonGLM08})
of a symmetric tensor also reduces to an equivalence problem involving
a particular multivariate
polynomial~(\cite{DBLP:conf/stoc/Kayal12}). To mention another
fundamental problem, the task of minimizing the cost of computing
matrix multiplication reduces to a particular equivalence
problem~(\cite{DBLP:conf/stoc/BurgisserI11,
  DBLP:conf/stoc/BurgisserI13,DBLP:conf/soda/CohnU13,DBLP:conf/stoc/Kayal12}).

\subsection*{Organization of the Paper and Main Results}
Let $\K$ be a field, $\f$ and $\g$ be two sets of $m$ polynomials over
$\K[x_1,\ldots,x_n]$. The Isomorphism of Polynomials (\IP) problem,
introduced by
Patarin~(\cite{DBLP:conf/eurocrypt/Patarin96}), is as
follows:\\
\textbf{Isomorphism of Polynomials (\IP)}\\
\textbf{Input:} $\big( (\f=(f_1,\ldots,f_m)$ and
$\g=(g_1,\ldots,g_m)\big) \in \K[x_1,\ldots,x_n]^m \times \K[x_1,\ldots,x_n]^m$.\\
\textbf{Question:} Find -- if any -- $(A,B) \in\GL_n(\K) \times
\GL_m(\K)$ such that:
\[\g(\x)=B\cdot \f(A \cdot \x), \mbox{ with
  $\x=(x_1,\dots,x_n)^{\T}$}.\] While \IP is a fundamental problem
in multivariate cryptography, there are quite few algorithms, such
as~\cite{DBLP:conf/eurocrypt/PatarinGC98,DBLP:conf/eurocrypt/BouillaguetFV13,FaugerePerret06},
solving \IP. In particular, \cite{FaugerePerret06} proposed to solve
\IP by reducing it to a system of nonlinear equations whose
variables are the unknown coefficients of the matrices.
It was conjectured in~\cite{FaugerePerret06},
but never proved, that the corresponding system of nonlinear equations
can be solved in polynomial time as soon as the \IP instances
considered are not homogeneous. Indeed, by slicing of the
polynomials degree by degree, one can find equations in the
coefficients of the transformation allowing one to recover the
transformation. More recently,
\cite{DBLP:conf/eurocrypt/BouillaguetFV13} presented exponential (in
the number of variables $n$) algorithms for solving quadratic
homogeneous instances of \IP over finite fields.  This situation is
clearly unsatisfying, and suggests that an important open problem for
\IP is to identify large class of instances which can be solved in
(randomized) polynomial time.

An important special case of \IP is the \emph{\IP problem with one
  secret} (\IPS for short), where $B$ is the identity matrix.  From
a cryptographic point of view, the most natural case encountered for
equivalence problems is inhomogeneous polynomials with affine
transformations.  For \IPS, we show that such a case can be handled
in the same way as homogeneous instances with linear transformations
(see Proposition~\ref{prop:homo}). As a side remark, we mention that
there exist more efficient methods to handle the affine case;
typically by considering the homogeneous components,
see~\cite{FaugerePerret06}. However, homogenizing the instances allows
us to make the proofs simpler and cleaner. As such, we focus our
attention to solve \IPS for quadratic homogeneous forms.

When \(m=1\), the \IPS problem can be easily solved by computing a
reduced form of the input quadratic
forms. In~\cite{DBLP:conf/pkc/BouillaguetFFP11}, the authors present
an efficient heuristic algorithm for solving \IPS on quadratic
instances. However, the algorithm requires to compute a Gr\"obner
basis. So, its complexity could be exponential in the worst case.
More recently, \cite{MacariotRatPG2013} proposed a polynomial-time
algorithm for solving \IPS on quadratic instances with $m=2$ over
fields of any characteristic.  We consider here arbitrary $m>1$.

In computer algebra, a fundamental and related
problem is the simplification of a homogeneous polynomial system
$\f\in\K[\x]^m$. That is, compute $A\in\GL_n(\K)$ such that
$\g(\x)=\f(A\cdot\x)$ is easier to solve. In this setting,
\textsc{Ridge} algorithm
(see~\cite{BerthomieuHM2010,Hironaka1970,Giraud1972}) and
\textsc{MinVar} algorithm (see~\cite{Carlini05reducingthe,Kayal2011}) reduce to
the best the number of variables of the system. More generally, for a given
homogeneous polynomial system $\f$, the \emph{Functional Decomposition
  Problem} is the problem of computing $\mathbf{h}=(h_1,\ldots,h_s)$
homogeneous and $\g$ such that $\f(\x)=\g(\mathbf{h}(\x))$.

To simplify the presentation in this introduction, we mainly deal with
fields of characteristic $\not =2$. Results for fields of
characteristic $2$ are also given later in this paper. Now, we define
formally \IPS:
\begin{definition}
  Let $\big( \f=(f_1,\ldots,f_m),\g=(g_1,\ldots,g_m)\big) \in
  \K[x_1,\ldots,x_n]^m \times \K[x_1,\ldots,x_n]^m$. We shall say that
  $\f$ and $\g$ are equivalent, denoted $\f \sim \g$, if $\exists$ $A
  \in \GL_n(\K)$ such that:
  \[\g(\x)=\f(A \cdot \x).\]
  \IPS is then the problem of finding -- if any -- $A \in \GL_n(\K)$
  that makes $\g$ equivalent to $\f$ (\emph{i.e.}\ $A \in \GL_n(\K)$
  such that $\g(\x)=\f(A \cdot \x)$\big).
\end{definition}
In such a case, we present a randomized polynomial-time algorithm for
solving \IPS with quadratic polynomials. To do so, we show that such
a problem can be reduced to the variant of a classical problem of
representation theory over finite dimensional algebras.  In our
setting we need, as in the case $m=1$, to provide a canonical form of
the problem.

\subsubsection*{Canonical Form of \IPS}
Let $\big( \f=(f_1,\ldots,f_m),\g=(g_1,\ldots,g_m)\big) \in
\K[x_1,\ldots,x_n]^m \times \K[x_1,\ldots,x_n]^m$ be homogeneous
quadratic polynomials. Let $H_1,\ldots,H_m$ be the Hessian matrices
of $f_1,\ldots,f_m$ (resp. $H^{\prime}_1,\ldots,H^{\prime}_m$ be the
Hessian matrices of $g_1,\ldots,g_m$). Recall that the Hessian matrix
associated to a $f_i$ is defined as \(H_i=\left(\frac{\partial^2
    f_i}{\partial x_k \partial x_{\ell}}\right)_{k,\ell}\in
\mathbb{K}^{n\times n}\). Consequently, \IPS for quadratic forms is
equivalent to finding \(A \in \GL_n(\K)\) such that:
\begin{equation}\label{eq:orig}
  H_i^{\prime}=A^{\T}\cdot H_i \cdot A,   \text{ for all } i, 1 \leq i \leq m.
\end{equation}
Assuming $H_j$ is invertible, and thus so is $H_j'$, one has
$H_j'^{-1}=A^{-1}H_j^{-1}A^{-\T}$. Combining this with
equation~\eqref{eq:orig} yields $H_j'^{-1}H_i'=A^{-1}\cdot
H_j^{-1}H_i\cdot A$. If none of the $H_i$'s is invertible, then we
look for an invertible linear combination thereof. For this reason, we
assume all along this paper:
\begin{assumption}[Regularity assumption]~\label{as:regular} Let
  $\f=(f_1,\ldots,f_m)\in\K[x_1,\ldots,x_n]$. We assume that a linear
  combination over $\K$ of the quadratic forms $f_1,\ldots,f_m$ is not
  degenerate\footnote{We would like to thank G. Ivanyos for having
    pointed us this issue in a preliminary version of this paper.}. In
  particular, we assume that $|\K| > n$.
\end{assumption}
Taking as variables the entries of $A$, we can see
that~\eqref{eq:orig} naturally yields a nonlinear system of
equations. However, we show that one can essentially linearize
equations~\eqref{eq:orig}. To this end, we prove in
Section~\ref{s:normalization} that under Assumption~\ref{as:regular}
any quadratic homogeneous instance \IPS can be reduced, under a
randomized process, to a canonical form on which -- in particular --
all the quadratic forms are nondegenerate. We shall call these
instances \emph{regular}. More precisely:
\begin{theorem}
  \label{th:normal}
  Let $\K$ be a field of $\car\K\neq 2$. There exists a randomized
  polynomial-time algorithm which given a regular quadratic
  homogeneous instance of \IPS returns ``\NoSol'' only if the two
  systems are not equivalent or a canonical form
  \[\pare{\big(\sum_{i=1}^{n}d_ix_{i}^{2},f_2,\ldots,f_m\big),
    \big(\sum_{i=1}^{n}d_ix_{i}^{2},g_2,\ldots,g_m\big)},\]
  where the $d_i$ are equal to $1$ or a nonsquare in $\K$, \(f_i\) and
  \(g_i\) are \emph{nondegenerate} homogeneous quadratic polynomials
  in \(\K[x_1,\ldots,x_n]\). Any solution on $\K$ on the canonical
  form can be efficiently mapped to a solution of the initial instance
  (and conversely).
\end{theorem}
Let us note that over the rationals, computing the exact same sum of
squares for the first quadratic forms of each set is
difficult, see~\cite[Chapter~3]{Saxena2006},
\cite[Chapter~1]{Wallenborn2013}.
As such, one could only
assume that the first quadratic form of the second set is
$\sum_{i=1}^nd_i'x_i^2$. This does not fundamentally change the
algorithms presented in this paper, beside some matrices denoted by
$D$ which could be changed into $D'=\Diag(d_1',\ldots,d_n')$.

Note that the success probability of the algorithms presented here
will depend on the size of the field. If one looks for
$A\in\LL^{n\times n}$ with $\LL$ an extension of $\K$, one can amplify
the success probability over a small field by using the fact that
matrices are conjugate over $\K$ if and only if they are conjugate
over an algebraic extension $\LL$ (see~\cite{Pazzis2010}). Thus, one
can search linear change of variables with coefficients in some
algebraic extension \(\LL\supseteq\K\) (but of limited
degree).

\subsubsection*{Conjugacy Problem}
When \IPS is given in canonical form, equations~\eqref{eq:orig} can
be rewritten as $A^{\T}\,D\, A = D$ with $D=\Diag(d_1,\ldots,d_n)$ and
$H_i^{\prime}=A^{\T}\cdot H_i \cdot A=D\,A^{-1}\,D^{-1}\cdot H_i \cdot
A$ for all $i, 2 \leq i \leq m$. Our task is now to solve the
following problem:
\begin{definition}[\textbf{$D$-Orthogonal Simultaneous Matrix Conjugacy
    (\MI)}]
  Let $\K^{n\times n}$ be the set of $n\times n$ matrices with entries
  in $\K$. Let $\{H_1,\ldots,H_m\}$ and
  $\{H^{\prime}_1,\ldots,H^{\prime}_m\}$ be two families of matrices
  in $\K^{n\times n}$. The \MI problem is the task to recover -- if
  any -- a $D$-orthogonal matrix $X \in \LL^{n\times n}$, \emph{i.e.}\
  $X^{\T}DX=D$, with \(\LL\) being an algebraic extension of \(\K\),
  such that:
  \begin{eqnarray*}
    X^{-1}\, H_i \, X=H^{\prime}_i, & \quad \forall\,i, 1 \leq i
    \leq m,\end{eqnarray*}  
\end{definition}
\cite{DBLP:conf/issac/ChistovIK97} show that \MI with $D=\id$ is
equivalent to:
\begin{enumerate}
\item Solving the Simultaneous Matrix\ Conjugacy problem (\SMC)
  between $\{H_{i}\}_{1 \leq i \leq m}$ and $\{H^{\prime}_i\}_{1 \leq
    i \leq m}$, that is to say finding an invertible matrix $Y \in
  \GL_n(\K)$ such that:
  \begin{eqnarray}\label{eq:smc0} Y^{-1}\cdot H_i \cdot Y=H^{\prime}_i
    & \text{ and } & Y^{-1}\cdot H_i^{\T} \cdot Y={H^{\prime}_i}^{\T}
    \quad \forall\,i, 1 \leq i \leq m.
  \end{eqnarray}  
\item Computing the square-root \(W\) of the matrix $Z=Y \cdot
  Y^{\T}$. Then, the solution of the \MI problem is given by \(X=Y\,
  W^{-1}\).
\end{enumerate}
In our context, $D=\Diag(d_1,\ldots,d_n)$ is any diagonal invertible
matrix. So, we extend~\cite{DBLP:conf/issac/ChistovIK97} and show that
\MI is equivalent to
\begin{enumerate}
\item Finding an invertible matrix $Y \in \GL_n(\K)$ such that:
  \begin{eqnarray} \label{eq:smc} Y^{-1}\cdot H_i \cdot Y=H^{\prime}_i
    & \text{ and } &D\,Y^{-1}D^{-1}\cdot H_i^{\T} \cdot
    D\,YD^{-1}={H^{\prime}_i}^{\T} \quad \forall\,i, 1 \leq i \leq m.
  \end{eqnarray}  
\item Computing the square-root \(W\) of the matrix $Z=D\,Y \cdot
  Y^{\T}D^{-1}$. Then, the solution of the \MI problem is given by
  \(X=Y\, W^{-1}\).
\end{enumerate}

In our case, the $H_i$'s (resp. $H^{\prime}_i$'s) are symmetric
(Hessian matrices). Thus, condition~\eqref{eq:smc} yields a system of
\emph{linear} equations and one polynomial inequation:
\begin{equation}\label{eq:lin}
  H_1\cdot Y=Y\cdot H^{\prime}_1,\ldots,H_m \cdot Y = Y\cdot H^{\prime}_m
  \text{ and }\det(Y)\neq 0.
\end{equation}
From now on, we shall denote by $\mcal{O}_n(\LL, D)$ the set of
$D$-orthogonal matrices with coefficients in $\LL$.

Let $V \subset \K^{n \times n}$ be the linear subspace of matrices
defined by these linear equations. The \SMC problem is then
equivalent to recovering an invertible matrix in $V$; in other words
we have to solve a particular instance of Edmonds'
problem~(\cite{edm67}). Note that, if the representation of the
algebra spanned by \(\{H_{1}^{-1}H_{i}\}_{1\leq i\leq m} \) is
\emph{irreducible}, we know that \(V\) has dimension at most \(1\)
(Schur's Lemma, see~\cite[Chap.~XVII, Proposition~1.1]{Lang2002}
and~\cite[Lemma~2]{Newman1967} for a matrix version of this
lemma). After putting the equations in triangular form, randomly
sampling over the free variables an element in $V$ yields, thanks to
Schwartz-Zippel-DeMillo-Lipton Lemma~(\cite{DL78,Z79}), a solution to
\MI with overwhelming probability as soon as $\K$ is big enough. If
one accepts to have a solution matrix over an extension field, we can
amplify the probability of success by considering a bigger algebraic
extension (see~\cite{Pazzis2010}).  Whilst a rather ``easy''
randomized polynomial-time algorithm solves \SMC, the task of finding
a deterministic algorithm is more delicate.  In our particular case,
we can adapt the result of~\cite{DBLP:conf/issac/ChistovIK97} and
provide a deterministic polynomial-time algorithm for
solving~\eqref{eq:smc0}.

\subsubsection*{Characteristic $2$}
Let us recall that in characteristic $2$, the associated matrices
$H_1,\ldots,H_m,H_1',\ldots,H_m'$ to quadratic forms can be chosen as
upper triangular. In this context, we show in Section~\ref{ss:char_2}
that \IPS can still be reduced to a $(H_1+H_1^{\T})$-conjugacy
problem.  Under certain conditions in even dimension, we can solve
this conjugacy problem in polynomial-time. These results are well
confirmed by some experimental results presented in
Section~\ref{ss:bench}.  We can recover a solution in less than one
second for $n$ up to one hundred (cryptographic applications of \IPS
usually require smaller values of $n$, typically $\leq 30$, for
efficiency reasons).

\subsubsection*{Matrix Square Root Computation}
It is well known that computing square roots of matrices can be done
efficiently using numerical methods (for instance,
see~\cite{grant}). On the other hand, it seems difficult to control
the bit complexity of numerical methods.
In~\cite[Section~3]{DBLP:conf/issac/ChistovIK97}, the authors consider
the problem of computing, in an exact way, the square root of matrices
over algebraic number fields. As presented, it is not completely clear
that the method proposed is polynomial-time as some coefficients of
the result matrix lie in extensions of nonpolynomial size,
see~\cite{Cai1994}. However, by applying a small trick to the proof
of~\cite{DBLP:conf/issac/ChistovIK97}, one can compute a solution in
polynomial-time for various field of characteristic $\neq 2$.
In any case, for the sake of completeness, we propose two
polynomial-time algorithms for this task. First, a general method
fixing the issue encountered
in~\cite[Section~3]{DBLP:conf/issac/ChistovIK97} is presented in
Section~\ref{ss:computation}.  To do so, we adapt the technique
of~\cite{Cai1994} and compute the square root as the product of two
matrices in an algebraic extension which can both be computed in
polynomial time. The delicate task being to control the size of the
algebraic extensions occurring during the algorithm. In here, each
coefficient of the two matrices are lying in an extension field of
polynomial degree in $n$. Furthermore, these matrices allow us to test
in polynomial time if $H_1,\ldots,H_m$ and $H_1',\ldots,H_m'$ are
indeed equivalent. We then present a
second simpler method based on the \emph{generalized Jordan normal
  form} (see Section~\ref{ss:sqrt_char_p}) which works (in polynomial
time) over finite fields. In general, it deals with algebraic
extensions of smaller degree than the first one. Putting things
together, we obtain our main result:
\begin{theorem}\label{th:probmain}
  Let $\K$ be a field with $\car\K\neq
  2$. Under Assumption~\ref{as:regular}, there exists a randomized
  polynomial-time algorithm for solving quadratic-\IPS over an
  extension field of $\K$ of polynomial degree in $n$.
\end{theorem}
Let us note
that the authentication scheme using \IPS requires to find a
solution over the base field. However, it is not always necessary to
find a solution in the base field (typically, in the context of a
key-recovery for multivariate schemes).  In~\cite{BettaleFP2013}, the
authors recover an equivalent key over an extension for the multi-HFE
scheme.

In addition, under some nondegeneracy assumption,
Theorem~\ref{th:probmain} can be turned into a deterministic algorithm
solving \IPS over the base field $\K$ or an extension thereof. That
is:
\begin{theorem}\label{th:main}
  Under Assumption~\ref{as:regular} and the assumption that one of the
  quadratic form is nondegenerate, there is a deterministic
  polynomial-time algorithm for solving quadratic-\IPS over an
  extension of $\K$ of polynomial degree in $n$. Furthermore, if the
  space of matrices satisfying equations~\eqref{eq:lin} has dimension $1$,
  then the algorithm can solve quadratic-\IPS over $\K$.
\end{theorem}
Let us note that assuming that one of the Hessian matrix is invertible
is not a strong assumption when the size of $\K$ is not too
small. Indeed, the probability of picking a random invertible
symmetric matrix over $\F_q$ is
\[\frac{\prod_{i=1}^{n}(1-q^{-i})}
{\prod_{i=1}^{\lfloor n/2\rfloor}(1-q^{-2i})}
=\prod_{i=1}^{\lceil n/2\rceil}(1-q^{-2i+1}),\]
see~\cite[Equations~4.7 and~4.8]{Carlitz1954}.

If $m\geq 3$, for random matrices $H_1,\ldots,H_m$, the set of
solutions of equations~\eqref{eq:lin} is a $1$-dimensional matrix
space. This allows us to solve quadratic-\IPS in polynomial-time
over $\K$.
In Section~\ref{ss:bench}, we present our
timings for solving \IPS. These experiments confirm that for randomly
chosen matrices, our method solves \IPS over $\K$. We remark also
that our method succeeds to solve \IPS over $\F_2$ for public-keys
whose sizes are much bigger than practical ones.

In Section~\ref{s:enumeration}, we consider the counting problem
\#\IPS associated to \IPS for quadratic (homogeneous) polynomials
in its canonical form (as defined in Theorem~\ref{th:normal}). Note
that such a counting problem is also related to cryptographic
concerns. It corresponds to evaluating the number of equivalent
secret-keys in \MPKC, see~\cite{DFPW11,WP11}. Given homogeneous quadratic
polynomials $\big( \f=(f_1,\ldots,f_m),\g=(g_1,\ldots,g_m)\big) \in
\K[x_1,\ldots,x_n]^m \times \K[x_1,\ldots,x_n]^m$, we want to count
the number of invertible matrices $A \in \GL_n(\K)$ such that $
\g(\x)=\f(A \cdot \x).$ To do so, we define:
\begin{definition}
  Let $\f=(f_1,\ldots,f_m) \in \K[x_1,\ldots,x_n]^m$, we shall call
  \emph{automorphism group of $\f$} the set:
  \[\mathcal{G}_{\f}=\{ A \in \GL_n(\K) \mid \f(A\cdot \x) =\f(\x)
  \}.\]
\end{definition}
If $\f \sim \g$, the automorphism groups of $\f$ and $\g$ are similar.
Thus, the size of the automorphism group of $\f$ allows us to count
the number of invertible matrices mapping $\f$ to $\g$. For quadratic
homogeneous polynomials, the automorphism group coincides with the
subset of regular matrices in the centralizer $\Cent(\cH)$ of the
Hessian matrices $\mcal{H}$ associated to $\f$. Taking $\alpha$ an
algebraic element of degree $m$ over $\K=\F_q$, let us assume the Jordan
normal form of $H=\sum_{i=i}^m\alpha^{i-1}\,H_i$ has Jordan blocks of
sizes $s_{i,1}\leq\cdots\leq s_{i,d_i}$ associated to eigenvalue
$\zeta_i$, for $i,1\leq i\leq r$. Then, as a consequence
of~\cite[Lemma~4.11]{Singla2010}, we prove that, if $q$ is an odd
prime power, then the number of
solutions of quadratic-\IPS in $\F_q^{n\times n}$ is bounded from
above by:
\[q^{\pare{\sum_{1\leq i\leq r}\sum_{1\leq j\leq d_i}(2\,d_i-2\,j+1)s_{i,j}}}-1.\]
 
\subsection*{Open Question: The Irregular Case}
Given a quadratic instance of \IPS, a nondegenerate instance is an
instance wherein the matrix whose rows are all the rows of
$H_1,\ldots,H_m$ has rank $n$.  In paragraph~\ref{par:variables}, we
see how to transform some degenerate instances into nondegenerate
instances. However, nondegenerate instances are not always regular
instances. There are cases, the so-called \emph{irregular} cases, such
that the vector space of matrices spanned by $H_1,\ldots,H_m$ does not
contain a nondegenerate matrix.  This situation is well illustrated by
the following example $f_1=x_1x_3,f_2=x_2x_3$.  Any linear combination
of $f_1,f_2$ is degenerate, while $\f=(f_1,f_2)$ is not.  Note that we
can decide in randomized polynomial time if an instance of
quadratic-\IPS is irregular since it is equivalent to checking if a
determinant is identically equal to zero; thus it is a particular
instance of polynomial identity testing.  In the irregular case, it is
clear that our algorithm fails.  In fact, it seems that most known
algorithms dedicated to
quadratic-\IPS~(\cite{DBLP:conf/pkc/BouillaguetFFP11,MacariotRatPG2013})
will fail
on these instances; making the hardness of the irregular case
intriguing and then an interesting open question.
 
\subsection*{Special case of \IP}
In our quest of finding instances of \IP solvable in
polynomial-time, we take a first step in Section~\ref{s:ip}.  We
consider \IP for a specific set of polynomials with $m=n$. In the
aforementioned Section~\ref{s:ip}, we prove the following:
\begin{theorem} \label{th:ip}
  Let $\g=(g_1,\ldots,g_n) \in \K[x_1,\ldots,x_n]^n$ be given in
  dense representation, and
  $\f=\POW_{n,d}=(x_1^d,\ldots,x_n^d) \in \K[x_1,\ldots,x_n]^n$ for
  some $d>0$. Whenever $\car\K=0$, let $e=d$ and
  $\tilde{\g}=\g$. Otherwise, let $p=\car\K$, let $e$ and $r$ be integers such
  that $d=p^r\,e$, $p$ and $e$ coprime, and let $\tilde{\g}\in\K[\x]^n$
  be such that $\g(\x)=\tilde{\g}(\x^{p^r})$. Let $L$ be a polynomial
  size of an arithmetic circuit to evaluate the determinant
  of the Jacobian matrix of $\tilde{\g}$. If the size
  of $\K$ is at least $12\,\max
  (2^{L+2},e\,(n-1)\,2^{e\,(n-1)}+e^3\,(n-1)^3,2\,(e\,(n-1)+1)^4)$,
  then there is a randomized
  polynomial-time algorithm which recovers -- if any --
  $(A,B)\in\GL_n(\K)\times\GL_n(\K)$ such that:
  \[\g=B\cdot\POW_{n,d}(A\cdot\x).\]
\end{theorem}
This extends a similar result of~\cite[Section~5]{Kayal2011} who
considered \PE for a sum of $d$-power polynomials. We show that
solving \IP for $\POW_{n,d}$ reduces to factoring the determinant of
a Jacobian matrix (in~\cite{Kayal2011}, the Hessian matrix is
considered). This illustrates, how powerful partial derivatives can be
in equivalence
problems~(\cite{DBLP:journals/fttcs/ChenKW11,DBLP:conf/eurocrypt/Perret05}).
To go along with the proof of Theorem~\ref{th:ip}, we design
Algorithm~\ref{al:ip} at the end of Section~\ref{s:ip}.

\section{Normalization - Canonical form of \IPS}
\label{s:normalization}
\setcounter{secnumdepth}{4}
\renewcommand\theparagraph{\thesection.\roman{paragraph}}
In this
section, we prove Theorem~\ref{th:normal}. In other words, we explain
how to reduce, under Assumption~\ref{as:regular}, any quadratic homogeneous
instance $(\f,\g) \in
\K[x_1,\ldots,x_n]^m \times \K[x_1,\ldots,x_n]^m$ of \IPS to a
suitable canonical form, \emph{i.e.}\ an instance of \IPS where all
the Hessian matrices are invertible and the first two equal
the same diagonal invertible matrix. We emphasize that the reduction presented
is randomized.

\paragraph{Homogenization}\label{par:homo}
We show here that the equivalence problem over inhomogeneous
polynomials with affine transformation on the variables reduces to the
equivalence problem over homogeneous polynomials with linear
transformation on the variables.  To do so, we simply homogenize the
polynomials.
Let $x_0$ be a new variable. For any polynomial $p\in\K[\x]$ of degree
$2$, we denote by
$p^{\star}(x_0,x_1,\ldots,x_n)=x_0^2\,p(x_1/x_0,\ldots, x_n/x_0)$ its
\emph{homogenization}.
\begin{proposition} \label{prop:homo} $\IPS$ with quadratic
  polynomials and affine transformation on the variables can be reduced in polynomial-time to \IPS with homogeneous quadratic polynomials and linear
  transformation on the variables.
\end{proposition}
\begin{proof}
  Let $(\f,\g)\in\K[\x]^m \times \K[\x]^m$ be inhomogeneous
  polynomials of degree $2$. We consider the transformation which maps
  $(\f,\g)$ to
  $\big(\f^{\star}=(f_0^{\star}=x_0^2,f_1^{\star},\ldots,f_m^{\star}),
  \g^{\star}=(g_0^{\star}=x_0^2,g_1^{\star},\ldots,g_m^{\star})\big)$. This
  clearly transforms polynomials of degree $2$ to homogeneous
  quadratic polynomials.  We can write
  $f_i(\x)=\x^{\T}\,H_i\,\x+L_i\,\x+c_i$ with $H_i\in\K^{n\times n}$,
  $L_i\in\K^n$ and $c_i\in\K$, then
  $f_i(A\x+b)=(A\x+b)^{\T}\,H_i\,(A\x+b)+L_i\,(A\x+b)+c_i$ and its
  homogenization is
  $(A\x+bx_0)^{\T}\,H_i\,(A\x+bx_0)+L_i\,(A\x+bx_0)\,x_0+c_i\,x_0^2
  =f^{\star}_i(A'\x^{\star})$,
  with $\x^{\star}=(x_0,x_1,\ldots,x_n)^{\T}$.  If $(A,b) \in
  \GL_{n}(\K) \times \K^n$ is an affine transformation solution on the
  inhomogeneous instance then $A^{\prime}=\pare{\begin{smallmatrix}1 &\mbf{0}\\
      b &A\end{smallmatrix}}$ is a solution for the homogenized
  instance. Conversely, a solution $A' \in \GL_{n+1}(\K)$ of the
  homogeneous problem must stabilize the homogenization variable $x_0$
  in order to be a solution of the inhomogeneous problem. This is
  forced by adding $f_0=x_0^2$ and $g_0=x_0^2$ and setting
  $C'=A'/a_{0,0}'$, with $a_{0,0}'=\pm 1$. One can see that $C'$ is of the form
  $\pare{\begin{smallmatrix}1 &\mbf{0}\\d &C\end{smallmatrix}}$, and
  $(C,d)\in \GL_n(\K) \times \K^n$ is a solution for $(\f,\g)$.
\end{proof}

\paragraph{Redundant Variables}\label{par:variables}
As a first preliminary natural manipulation, we first want to
eliminate -- if any -- \emph{redundant variables} from the instances
considered.  Thanks to~\cite{Carlini05reducingthe} (and
reformulated in~\cite{Kayal2011}), this task can be done in
randomized polynomial time:
\begin{proposition}(\cite{Carlini05reducingthe,Kayal2011})
  \label{prop:carl}
  Let $f \in \K[x_1\ldots,x_n]$ be a polynomial. We shall say that $f$ has $s$
  \emph{essential variables} if $\exists \,M \in \GL_n(\K)$ such that $f(M
  \x)$ depends only on the first $s$ variables $x_1,\ldots,x_s$. The
  remaining $n-s$ variables $x_{s+1}\ldots,x_n$ will be called
  \emph{redundant} variables.
  If $\car\K=0$ or
  $\car\K>\deg f$, and $f$ has $s$ essential variables, then we can
  compute in randomized polynomial time $M \in \GL_n(\K)$ such that $f(M\,\x)$
  depends only on the first $s$ variables.
\end{proposition}
For a quadratic form, $s$ is simply the rank of the associated Hessian
matrix. As such, for $m=1$, a quadratic instance is regular if and
only if the associated Hessian matrix is invertible.
For a set of equations, we extend the notion of essential variables as
follows.
\begin{definition}\label{def:ess}
  The number of \emph{essential variables} of
  $\f=(f_1,\ldots,f_m)\in\K[x_1,\ldots,x_n]^m$ is the smallest $s$
  such that $\f$ can be decomposed as:
  \[\f=\tilde{\f}(\ell_1,\ldots,\ell_{s})\]
  with $\ell_1,\ldots,\ell_{s}$ being linear forms in $x_1,\ldots,x_n$ of
  rank $s$ and $\tilde{\f}\in\K[y_1,\ldots,y_{s}]^m$.
\end{definition}
The linear forms $\ell_1,\ldots,\ell_s$ can be easily
computed thanks to Proposition~\ref{prop:carl} when the characteristic of
$\K$ is zero or greater than the degrees of $f_1,\ldots,f_m$. In
characteristic $2$, when $\K$ is perfect (which is always true if $\K$
is finite for instance) the linear forms can also be recovered in
polynomial time (see~\cite{BerthomieuHM2010,Giraud1972,Hironaka1970} for
instance).  Below, we show that we can restrict our attention to only
essential variables. Namely, solving \IPS on $(\f,\g)$ reduces to
solving \IPS on instances having only essential variables.
\begin{proposition} \label{lm:rank} Let
  $(\f,\g)\in\K[x_1,\ldots,x_n]^m \times \K[x_1,\ldots,x_n]^m $ be two
  sets of quadratic polynomials. If
  $\f \sim \g$, then their numbers of essential variables must be the
  same. Let $s$ be the number of essential variables of $\f$. Finally,
  let $(\tilde{\f},\tilde{\g})\in\K[y_1,\ldots,y_s]^m \times
  \K[y_1,\ldots,y_s]^m$ be such that:
  \[\f=\tilde{\f}(\ell_1,\ldots,\ell_{s}) \mbox{ and }
  \g=\tilde{\g}(\ell'_1,\ldots,\ell'_{s}),\] with
  $\ell_1,\ldots,\ell_{s}$ (resp. $\ell'_1,\ldots,\ell'_{s}$) linear
  forms in $\x$ of rank $s$ and $\tilde{\f},
  \tilde{\g}\in\K[y_1,\ldots,y_{s}]^m$.  It holds that:
  \[\f \sim \g \iff \tilde{\f} \sim \tilde{\g}.\]
\end{proposition}
\begin{proof}
  Let $H_1,\ldots,H_m$ be the Hessian matrices of
  $f_1,\ldots,f_m$ (resp. $H^{\prime}_1,\ldots,H^{\prime}_m$ be the
  Hessian matrices of $g_1,\ldots,g_m$). Similarly, we define the
  Hessian matrices $\tilde{H}_1,\ldots,\tilde{H}_m$
  (resp. $\tilde{H}^{\prime}_1,\ldots,\tilde{H}^{\prime}_m$)
  of $\tilde{f}_1,\ldots,\tilde{f}_m$
  (resp. $\tilde{g}_1,\ldots,\tilde{g}_m$). Let also $M$ and $N$ be
  matrices in $\GL_n(\K)$ such that
  $H_i=M^{\T}\pare{\begin{smallmatrix} \tilde{H}_i & \mbf{0}\\\mbf{0}
      &\mbf{0}\end{smallmatrix}}M$ and $H_i'=N^{\T}\pare{\begin{smallmatrix}
      \tilde{H}_i' &\mbf{0}\\\mbf{0} &\mbf{0}\end{smallmatrix}}N$ for all $i,1\leq i\leq
  m$. There exist such $M$ and $N$, as $\f$ and $\g$ have essentially
  $s$ variables.  Up to re-indexing the rows and columns of $H_i$ and $H_i'$, so
  that they remain symmetric, one can
  always choose $M$ and $N$ such that
  $M=\pare{\begin{smallmatrix}M_1&M_2\\\mbf{0}&\id\end{smallmatrix}}$ and
  $N=\pare{\begin{smallmatrix}N_1&N_2\\\mbf{0}&\id\end{smallmatrix}}$, with
  $M_1,N_1\in\GL_s(\K)$.

  If $\tilde{\f} \sim \tilde{\g}$, $\exists \, \tilde{A}\in\GL_s(\K)$
  such that $A^{\T}\tilde{H}_i\tilde{A}=\tilde{H}_i'$, for all
  $i,1\leq i\leq m$. Then, for all $B\in\K^{(n-s)\times s}$ and $C\in\GL_{n-s}(\K)$:
  \begin{align*}
    \pare{\begin{smallmatrix}\tilde{A}^{\T} &B^{\T}\\ \mbf{0}
        &C^{\T}\end{smallmatrix}}
  \pare{\begin{smallmatrix}\tilde{H}_i\vphantom{^{\T}}&\mbf{0}\\\mbf{0} &\mbf{0}\vphantom{^{\T}}
    \end{smallmatrix}}
  \pare{\begin{smallmatrix}\tilde{A}\vphantom{^{\T}} &\mbf{0}\\ B &C\vphantom{^{\T}}
    \end{smallmatrix}}
  &=
  \pare{\begin{smallmatrix}\tilde{H}_i^{\prime}\vphantom{^{\T}}&\mbf{0}\\\mbf{0}
      &\mbf{0}\vphantom{^{\T}}
    \end{smallmatrix}},\\
  N^{\T}\pare{\begin{smallmatrix}\tilde{A}^{\T} &B^{\T}\\ \mbf{0}
      &C^{\T}\end{smallmatrix}} M^{-\T}H_i\,M^{-1}
  \pare{\begin{smallmatrix}\tilde{A}\vphantom{^{\T}} &\mbf{0}\\ B
      &C\vphantom{^{\T}}
    \end{smallmatrix}}N &=H_i^{\prime}.
\end{align*}
Therefore, $\f$ and $\g$ are equivalent.

Conversely, we assume now that $\f \sim \g$, \emph{i.e.}\ there exists
$A\in\GL_n(\K)$ such that $A^{\T}\cdot H_i \cdot A=H^{\prime}_i$, for
all $i,1 \leq i \leq m$. This implies that:
\[
N^{-\T}A^{\T}M^{\T}\pare{\begin{smallmatrix}\tilde{H}_i &\mbf{0}\\\mbf{0} &\mbf{0}
  \end{smallmatrix}}M\,A\,N^{-1}=\pare{\begin{smallmatrix}\tilde{H}^{\prime}_i
    &\mbf{0}\\\mbf{0} &\mbf{0}
  \end{smallmatrix}}, \forall\, i,1 \leq i \leq m.
\]
We then define $\tilde{A}=((MAN^{-1})_{i,j})_{1\leq i,j\leq s}$, so that
$\tilde{\f}(\tilde{A}\x)=\tilde{\g}(\x)$. As $\g$ has $s$ essential
variables, then $\rank\tilde{A}$ cannot be smaller than $s$, hence
$\tilde{A} \in \GL_s(\K)$. We then get
$\tilde{A}^{\T}\tilde{H}_i\tilde{A}=\tilde{H}_i'$ for all $i,1 \leq i
\leq m$, \emph{i.e.}\ $\tilde{\f} \sim \tilde{\g}$.
\end{proof}
According to Proposition~\ref{lm:rank}, there is an efficient reduction
mapping an instance $(\f,\g)$
of \IPS to an instance $(\tilde{\f},\tilde{\g})$ of \IPS having
only essential variables. From now on, we will then assume that
we consider instances of \IPS with $n$ essential
variables for both $\f$ and $\g$.

\paragraph{Canonical Form}\label{par:identity}
We now assume that $\car\K\neq 2$. 
\begin{definition}
 Let $\f=(f_1,\ldots,f_m)\in\K[x_1,\ldots,x_n]^m$ be quadratic
  homogeneous forms with Hessian matrices $H_1,\ldots,H_m$.
We shall say that $\f$ is \emph{regular} if its number of essential variables is
$n$ and if $\exists\,\lambda_1,\dots,\lambda_m\in\K$ such
that $\det\pare{\sum_{i=1}^m\lambda_i\,H_i}\neq 0$.
\end{definition}
\begin{remark}
Our algorithm requires   
that amongst all the Hessian matrices, one at least is invertible, the so-called
regular case.
It is not sufficient to only assume that the number of essential variables is
$n$.  Indeed, Ivanyos's irregular example $\f=(x_1x_3,x_2x_3)$ has $3$ essential
variables, but any nonzero linear combination $\lambda_1 f_1+\lambda_2 f_2$ has
only $2$ essential variables $\lambda_1 x_1+\lambda_2 x_2$ and $x_3$.
Similarly,
$\f=(x_1^2+x_2^2+x_3^2,x_2^2+2x_3^2+x_4^2)$ has $4$ essential variables but any
nonzero linear combination $\lambda_1f_1+\lambda_2f_2$ over $\F_3$ has only
$3$ essential variables. This explains the additional condition on the previous
definition, and our Assumption~\ref{as:regular}.  
\end{remark}
We are now in a position to reduce quadratic homogeneous instances of
\IPS to a first simplified form.
\begin{proposition} \label{prop:first_reduction} Let
  $(\f,\g)\in\K[x_1,\ldots,x_n]^m \times \K[x_1,\ldots,x_n]^m $ be regular
  quadratic homogeneous polynomials.
  There is a randomized polynomial-time algorithm which returns
  ``\NoSol'' only if $\f \not \sim \g$, or a new instance
  \[
  (\tilde{\f},\tilde{\g})=\pare{\big(\sum_{i=1}^{n}d_ix_{i}^{2},\tilde{f}_2,
    \ldots,\tilde{f}_m\big),\big(\sum_{i=1}^{n}d_ix_{i}^{2},\tilde{g}_2,\ldots,
    \tilde{g}_m\big)} \in \K[\x]^m \times \K[\x]^m,\]
  with
  $d_1,\ldots,d_n$ being $1$ or nonsquares in $\K$, such that
  $\f\sim\g \iff \tilde{\f} \sim \tilde{\g}$.
  If $\K$ is finite, the
  output of this algorithm is correct with probability at least
  $1-n/|\K|$.
  If $\tilde{\f} \sim \tilde{\g}$, invertible matrices $P, Q$ and
  $A' \in \GL_n(\K)$
  are returned such that $\f(P\x)=\tilde{\f}(\x)$,
  $\g(Q\x)=\tilde{\g}(\x)$ and
  $\tilde{\f}(A^{\prime}\x)=\tilde{\g}(\x)$. It then holds that
  $\f(PA'Q^{-1}\x)=\g(\x)$.
\end{proposition}
\begin{proof}
  Let $H_1,\ldots,H_m$ be the Hessian matrices associated to $f_1,\ldots,f_m$.
  According to Schwartz-Zippel-DeMillo-Lipton Lemma~(\cite{DL78,Z79}),
  we can compute
  in randomized polynomial time $\lambda_1,\ldots,\lambda_m \in \K$
  such that $\vphi=\sum_{i=1}^{m}\lambda_i \cdot f_i$ is regular,
  \emph{i.e.}\ $\det\pare{\sum_{i=1}^m\lambda_i\,H_i}\neq 0$.
  The probability to pick $(\lambda_1,\ldots,\lambda_m)\in\K^m$ on which
  $\vphi$ is regular is bounded from above by $n/|\K|$. 
  We define $\gamma=\sum_{i=1}^{m}\lambda_i \cdot g_i$.  Should one reorder
  the equations, we can assume w.lo.g.\ that $\lambda_1\neq 0$. We
  have then:
  \[\f \sim \g \iff (\vphi,f_2,\ldots,f_m) \sim (\gamma,g_2,\ldots,g_m).\]
  Now, applying \Gauss's reduction algorithm to $\vphi$, there exists
  $d_1,\ldots,d_n\in\K$, each being $1$ or a nonsquare, such that
  $\vphi=\sum_{i=1}^n{d_i\,\ell_i^2}$, where
  $\ell_1,\ldots,\ell_n$ are independent linear forms in
  $x_1,\ldots,x_n$.
  This gives a $P \in\GL_n(\LL)$ such that
  $\tilde{\f}=(\tilde{\vphi}=\sum_{i=1}^nd_i\,x_i^2,\tilde{f_2},\ldots,
  \tilde{f_m})=(\vphi(P\x),f_2(P\x),\ldots,f_m(P\x))$.
  Clearly, $\f\sim\tilde{\f}$, hence, $\tilde{\f}\sim\g$.

  After that, we can apply once again \Gauss's reduction algorithm to
  $\gamma$. If the reduced polynomial  is different from
  $\sum_{i=1}^{n}d_i\,x_i^2$, then $\f \not \sim \g$ and we return
  ``\NoSol''. Otherwise, the reduction is given by a matrix $Q \in
  \GL_n(\LL)$ such that
  $\tilde{\g}=(\tilde{\gamma}=\sum_{i=1}^nd_i\,x_i^2,\tilde{g}_2,
  \ldots\tilde{g}_m)=
  (\gamma(Q\x),g_2(Q\x),\ldots,g_m(Q\x))$ and
  $\g\sim\tilde{\g}$. Thus, $\tilde{\f}\sim\tilde{\g}$ if
  and only if $\f\sim\g$.

  Now, assume that $ \exists\,A^{\prime} \in
  \GL_n(\K)$ such that $\tilde{\f}(A^{\prime}\x)=\tilde{\g}(\x)$. Then,
  $\f(PA^{\prime}\x)=\g(Q\x)$, \emph{i.e.}\ 
  $\f(PA^{\prime}Q^{-1}\x)=\g(\x)$.
\end{proof}
Let us recall that whenever $\K=\Q$, computing the exact same sum of
squares for $\tilde{f}_1$ and $\tilde{g}_1$ is
difficult, see~\cite[Chapter~3]{Saxena2006},
\cite[Chapter~1]{Wallenborn2013}. As such, we could only assume that
our canonical form is $\tilde{g}_1=\sum_{i=1}^nd_i'x_i^2$. This would
merely change the formulation of following
Theorem~\ref{th:second_reduction}.

\paragraph{Invertible Hessian Matrices}\label{par:invertible}
We are now in a position to reduce any regular homogeneous quadratic instances
$(\f,\g)$ of \IPS to a new form of the instances where all the
polynomials are themselves regular
assuming we could find one. From
Proposition~\ref{prop:first_reduction},
this is already the case -- under randomized reduction -- for $f_1$
and thus $g_1$.  For the other polynomials, we proceed as
follows. For $i,2\leq i\leq m$, if the Hessian matrix $H_i$ of $f_i$
is invertible, then we do nothing. Otherwise, we change $H_i$ into
$H_i-\nu_i\,H_1$, with $\nu_i$ not an eigenvalue of $H_i\,H_1^{-1}$.
As $\K$ has at least $n+1$ elements, there exists
such a $\nu_i$ in $\K$. This gives the following result:
\begin{theorem}\label{th:second_reduction}
  Let $(\f,\g)\in\K[x_1,\ldots,x_n]^m \times \K[x_1,\ldots,x_n]^m $ be regular
  quadratic homogeneous polynomials. There is a randomized
  polynomial-time algorithm which returns ``\NoSol'' only if $\f \not
  \sim \g$. Otherwise, the algorithm returns two sets of
  $n\times n$
  invertible symmetric matrices
  $\{D,\tilde{H}_2\ldots,\tilde{H}_m\}$ and
  $\{D,\tilde{H}^{\prime}_2,\ldots,\tilde{H}^{\prime}_m\}$, with $D$ diagonal,
  defined
  over $\K$
  such that:
  \[\g(\x)=\f(A\,\x), \text{ for } A \in \GL_n(\K)
  \iff
  \begin{array}{l}
    A'^{-1}\,D^{-1}\,\tilde{H}_i\,A^{\prime}=D^{-1}\,\tilde{H}_i',
    \forall \,i, 1\leq i\leq m,\\
    \text{ for } A^{\prime} \in \mcal{O}_n(\K,D),
  \end{array}
  \] with
  $\mcal{O}_n(\K,D)$ denoting the set of $n \times n$ $D$-orthogonal
  matrices over $\K$.
\end{theorem}
\begin{proof}
  Combining Proposition~\ref{prop:first_reduction} and
  paragraph~\ref{par:invertible} any regular quadratic homogeneous instance of
  \IPS can be reduced in randomized polynomial time to ``\NoSol'',
  only if the two systems are not equivalent, or to a
  \[(\tilde{\f},\tilde{\g})=\pare{\big(\sum_{i=1}^{n}d_ix_{i}^{2},\tilde{f}_2,
    \ldots,\tilde{f}_m\big),\big(\sum_{i=1}^{n}d_ix_{i}^{2},\tilde{g}_2,\ldots,
    \tilde{g}_m\big)},\]
  where all the polynomials are \emph{nondegenerate} homogeneous
  quadratic polynomials in \(\K[\x]\).
  It follows that
  $\tilde{\f}\sim\tilde{\g} \iff \exists\, A^{\prime} \in\GL_n(\K)$
  such that $\forall\,i, 1\leq i\leq m$,
  ${A^{\prime}}^{\T}\,\tilde{H}_i\,A^{\prime}=\tilde{H}_i'$.  In particular
  ${A^{\prime}}^{\T}\, D \, A^{\prime}=D$ and $A'$ is $D$-orthogonal. Hence,
  ${A^{\prime}}^{\T}\,\tilde{H}_i\,A^{\prime}=D\,A'^{-1}\,D^{-1}\,\tilde{H}_i\,A'=
  \tilde{H}_i',\ \forall \,i, 1\leq i\leq m$.
\end{proof}
The proof of this result implies Theorem~\ref{th:normal}.

\paragraph{Field Extensions and Jordan Normal Form}\label{par:jordan}
To amplify the success probability of our results, it will be
convenient to embed a field $\F$ in some finite extension $\F'$ of $\F$. This
is motivated by the fact that matrices in $\F^{n\times n}$ are similar
if and only if they are similar in $\F'^{n\times n}$,
see~\cite{Pazzis2010}. In this paper, we will need to compute
the \emph{Jordan normal form} $J$ of some matrix $H$ in several
situations. The computation of the Jordan normal form is done in two
steps.  First, we factor the characteristic polynomial, using for
instance Berlekamp's algorithm over
$\F=\F_q$ in $O(n\,\msf{M}(n)\log (q\,n))$ operations in $\F$, where
$\msf{M}(n)$ is a bound on the number of operations in $\F$ to
multiply two polynomials in $\F[x]$ of degree at most $n-1$,
see~\cite[Theorem~14.14]{GaGe1999}. Then, we
use~\cite{Storjohann1998}'s algorithm to compute the
generalized eigenvectors in $O(n^{\omega})$ operations in $\F$, with
$\omega$ being the exponent of time complexity of matrix
multiplication, $2\leq\omega\leq 3$.
\setcounter{secnumdepth}{3}

\section{Quadratic IP1S}\label{s:IP1S}
In this section, we present efficient algorithms for solving
regular quadratic-\IPS.  According to Proposition~\ref{prop:homo}, we can
w.l.o.g.\ restrict our attention on linear changes of variables and
homogeneous quadratic instances. 
Let $D$ be a diagonal
invertible matrix with $1$ or nonsquare elements on the diagonal. Let
$\mathcal{H}=\{D,H_2,\ldots,H_m\}$ and
$\mathcal{H}'=\{D,H^{\prime}_2,\ldots,H^{\prime}_m\}$ be two
families of invertible symmetric matrices in $\K^{n\times n}$. As
explained in Theorem~\ref{th:second_reduction}, our task reduces --
under a randomized process -- to finding a $D$-orthogonal matrix
$A^{\prime} \in \mcal{O}_n(\K,D)$ such that:
\begin{equation}\label{eq:main}
  A'^{-1}\,D^{-1}\,H_i\,A^{\prime}=D^{-1}\,H_i',\  \forall \,i, 1\leq i\leq m.
\end{equation}
Case $D=\id$ was studied in~\cite[Theorem~$4$]{DBLP:conf/issac/ChistovIK97}.
The authors prove
that there is an orthogonal solution $A$,
such that $H_i\,A=A\,H_i'$ if
and only if there is an invertible matrix $Y$
such that $H_i\,Y=Y\,H_i'$ and
$H_i^{\T}\,Y=Y\,H_i^{\T}$. In our case, whenever $D=\id$, the matrices
are symmetric.  So,
the added conditions -- with the transpose -- are automatically
fulfilled. In~\cite{DBLP:conf/issac/ChistovIK97}, the authors suggest
then to use the polar decomposition of $Y=A\,W$, with $W$ symmetric and
$A$ orthogonal. Then, $A$ is an orthogonal solution of~\eqref{eq:main}.

The main idea to compute $A$ is to compute $W$ as the square root of
$Z=Y^{\T}\,Y$ as stated
in~\cite[Section~3]{DBLP:conf/issac/ChistovIK97}. However, in general
$W$ and $A$ are not defined over $\K$ but over
$\LL=\K(\sqrt{\zeta_1},\ldots,\sqrt{\zeta_r})$, where
$\zeta_1,\ldots,\zeta_r$ are the
eigenvalues of $Z$. Assuming $\zeta_1$ is the root of an irreducible
polynomial $P$ of degree $d$, then $\zeta_2,\ldots,\zeta_d$ are also
roots of the same polynomial. However, there is no reason for them to
be in $\K[x]/(P)=\K(\zeta_1)$. But they will be the roots of a
polynomial of degree $d-1$, in general, over the field
$\K(\zeta_1)$. Then, doing another extension might only add one
eigenvalue in the field. Repeating this process yields a field of
degree $d!$ over $\K$. As a consequence, in the worst case, we can
have to work over an extension field of degree $n!$. Therefore,
computing $W$ could be the bottleneck of the method.

\cite{DBLP:conf/issac/ChistovIK97} emphasize
that constructing such a square root $W$ in polynomial time is the
only serious algorithmic problem.  As presented, it is not completely
clear that the method proposed is efficient.
They propose to compute $W=\sqrt{Y^{\T}\,Y}$ and then to set
$A=W^{-1}\,Y$. According to Cai's work~(\cite{Cai1994}), some coefficients of
matrix $A$ may lie in an extension of exponential
degree. \emph{Blockwise computation} (see the proof of
Proposition~\ref{prop:comput_ortho}) can allow us to compute such a
matrix. Chistov, Ivanyos and Karpinski set $y_i$ as the
restriction of $Y$ to the $i$th eigenspace, associated to $\zeta_i$,
of $Y^{\T}\,Y$. Then, $x_i=\sqrt{\zeta_i}^{-1}y_i$ and they
return the block diagonal matrix constructed from the
$x_i$'s. However, this construction gives the impression that the
$i$th eigenspace of $Y^{\T}\,Y$ is stable by $Y$, as
$W^{-1}$ would act as a multiplication by $\sqrt{\zeta_i}^{-1}$. As
a consequence, the blockwise computation was not ensured.

However, this issue does not happen if one uses the same proof
on $W=\sqrt{Y\,Y^{\T}}$ and
$A=Y\,W^{-1}$. In the following subsection,
we extend their proof to any invertible diagonal matrix $D$.

\subsection{Existence of a $D$-Orthogonal Solution}
The classical polar decomposition is used
in~\cite[Theorem~4]{DBLP:conf/issac/ChistovIK97} to determine an
orthogonal solution. Using the analogous decomposition, the so-called
Generalized Polar Decomposition (\GPD), which depends on $D$,
yields a $D$-orthogonal solution, see~\cite{MaMaTi2005}.  The \GPD
of an invertible matrix $Y$ is the factorization $Y=A\,W$, with $A$
$D$-orthogonal and $W$ in the associated Jordan algebra,
\emph{i.e.}\ $W^{\T}=D\,W\,D^{-1}$. Let us notice that $A$ and $W$
might be defined only over $\K'$ an algebraic extension of $\K$ of
some degree.

\begin{proposition}\label{prop:pseudoortho_solution}
  Let $\cK=\{K_1,\ldots,K_m\}$ and $\cK'=\{K_1',\ldots,K_m'\}$
  be two subsets of $m$ matrices in $\K^{n\times n}$.
  Let $D$ be an invertible diagonal matrix. There is a
  $D$-orthogonal solution $A\in\K'^{n\times n}$ to the conjugacy problem
  $K_i\,A=A\,K_i'$ for all $1\leq i\leq m$, if and only if there is
  an invertible solution $Y\in\K'^{n\times n}$ to the conjugacy
  problem $K_i\,Y=Y\,K_i'$ and
  $K_i^{\T}\,D\,Y\,D^{-1}=D\,Y\,D^{-1}\,K_i'^{\T}$ for all
  $1\leq i\leq m$.
  Furthermore, if $Y=A\,W$ is the \GPD of $Y$ with respect to
  $D$, then $A$ suits.
\end{proposition}
\begin{proof}
  This proof is a generalization
  of~\cite[Section~3]{DBLP:conf/issac/ChistovIK97}. If $A$ is a
  $D$-orthogonal solution to the first problem, then as
  $A^{\T}=D\,A^{-1}\,D^{-1}$, it is clear that $A$ is a solution to the
  second problem. Conversely, let $Y$ be a solution to the second
  problem, then $Z=D^{-1}\,Y^{\T}\,D\,Y$ commutes with $K_i'$. As
  $Y$ is invertible, so is $Z$, therefore, given a determination of the
  square roots of the eigenvalues of $Z$, there is a unique matrix $W$
  with these eigenvalues such that $W^2=Z$ and $W$ is in the Jordan
  algebra associated to $D$, that is $W^{\T}=D\,W\,D^{-1}$,
  see~\cite[Theorem~6.2]{MaMaTi2005}. As such, $W$ is a polynomial
  in $Z$ as proven in Section~\ref{ss:sqrt} and commutes with $K_i'$.

  Finally, $A=Y\,W^{-1}$ is an $D$-orthogonal solution of the first
  problem. As $W$ commutes with $K_i'$,
  $A^{-1}K_i\,A=W\,Y^{-1}K_i\,Y\,W^{-1}=W\,K_i'\,W^{-1}=K_i'$ and
  \[A^{\T}\,D\,A=W^{-\T}\,Y^{\T}\,D\,Y\,W^{-1}=
  D\,W^{-1}\,D^{-1}\,Y^{\T}\,D\,Y\,W^{-1}=
  D\,W^{-1}\,Z\,W^{-1}=D.\qedhere\]
\end{proof}

For the sake of completeness,
we present several efficient algorithms for performing the square root
computation.
\subsection{Computing the $D$-Orthogonal
  Solution} \label{ss:computation} The goal of this part is to
``$D$-orthogonalize'' an invertible solution $Y \in\GL_n(\K)$ of
equation~\eqref{eq:main}.  Instead of computing exactly
$A \in \mcal{O}_n(\LL,D)$,
we compute in polynomial time two matrices whose product is $A$. These
matrices allow us to verify in polynomial time that $H_i$ and $H_i'$ are
equivalent for all $i, 1\leq i\leq m$. To be more precise, we prove the
following proposition.
\begin{proposition}\label{prop:comput_ortho}
  Let $\cH=\{H_1=D,H_2,\ldots,H_m\}$ and $\cH'=\{H_1'=D,H_2',\ldots,H_m'\}$ be
  two sets of invertible matrices in $\K^{n\times n}$. We can compute in
  polynomial time two matrices representations of matrices
  $S$ and $T$ defined over an algebraic
  extension $\LL$ such that $S\,T^{-1}$ is $D$-orthogonal and for all
  $1\leq i\leq m$, $D^{-1}\,H_i(S\,T^{-1})=(S\,T^{-1})\,D^{-1}\,H_i'$.
  In the worst case,
  product $S\,T^{-1}$ cannot be computable in polynomial time over
  $\LL$.  However, matrices $S^{\T}\,H_i\,S$ and $T^{\T}\,H_i'\,T$ can be
  computed and tested for equality in polynomial time.
\end{proposition}
\begin{proof}
  Let $Y\in\GL_n(\K)$ such that $D^{-1}\,H_i\,Y=Y\,D^{-1}\,H'_i,\ \forall \,i,
  1\leq i\leq m$.  We set $Z=D^{-1}\,Y^{\T}\,D\,Y$.  Let us denote by
  $T$, the change of basis matrix such that $J=T^{-1}\,Z\,T$ is the
  Jordan normal form of $Z$. According to~\cite{Cai1994}, $T$,
  $T^{-1}$ and $J$ can be computed in polynomial time. Because of the
  issue of mixing all the
  eigenvalues of $Z$, we cannot compute efficiently $A$ in one piece.
  We will then compute $A\,T$ and $T^{-1}$ separately. Indeed, $A\,T$
  (resp. $T^{-1}$) is such that each of its columns (resp. each of its
  rows) is defined over an extension field $\K(\zeta_i)$, where
  $\zeta_1,\ldots,\zeta_r$ are the eigenvalues of $Z$.

  We shall say that a matrix is \emph{block-wise} (resp.
  \emph{columnblock-wise}, \emph{rowblock-wise}) \emph{defined over
    $\K(\zeta)$} if for all $1\leq i\leq r$, its $i$th block
  (resp. block of columns, block of rows) is defined over
  $\K(\zeta_i)$. The size of the $i$th block being the size of the
  $i$th Jordan block of $J$.

  As $J=T^{-1}\,Z\,T$ is a Jordan normal form, it is block-wise defined
  over $\K(\zeta)$. Using the closed formula of
  Section~\ref{ss:sqrt}, one can compute in polynomial time a square
  root $G$ of $J$. This matrix is a block diagonal matrix, block-wise
  defined over $\K(\sqrt{\zeta})$, hence it can be inverted in
  polynomial time. Should one want $W$, one would have to compute
  $W=T\,G\,T^{-1}$. Let us recall that matrices $T$ and $T^{-1}$ are
  respectively columnblock-wise and rowblock-wise defined over
  $\K(\zeta)$, see~\cite[Section~4]{Cai1994}. Since $Y$ is defined
  over $\K$, then $Y\,T$ is columnblock-wise defined over
  $\K(\zeta)$. Thus $S=A\,T=Y\,W^{-1}\,T=Y\,T\,G^{-1}$ is columnblock-wise
  defined over $\K(\sqrt{\zeta})$. We recall that product $A\,T\cdot
  T^{-1}$ mangles the eigenvalues and make each coefficient defined
  over $\K(\sqrt{\zeta_1},\ldots,\sqrt{\zeta_r})$ and thus must
  be avoided.

  Now, to verify that $A^{\T}\,H\,A=H'$, for any $H\in\cH$ and the
  corresponding $H'\in\cH'$, we compute separately
  $S^{\T}\,H\,S=T^{\T}\,A^{\T}\,H\,A\,T$ and $T^{\T}\,H'\,T$.
  For the former, $S=A\,T$
  (resp. $S^{\T}=(A\,T)^{\T}$) is columnblock-wise (resp. rowblock-wise)
  defined over $\K(\sqrt{\zeta})$ and $H$ is defined over
  $\K$. Therefore, the product matrix makes each of the coefficients
  which are
  on both the $i$th block of rows and the $j$th block of columns
  defined over $\K(\sqrt{\zeta_i},\sqrt{\zeta_j})$ and so can be
  computed in polynomial time. For the latter, the same behaviour
  occurs on the resulting matrix as $T$ is columnblock-wise defined
  over $\K(\zeta)$.
\end{proof}

  Let us assume that the characteristic polynomial of $Z$, of degree
  $n$, can be factored as $P_1^{e_1}\cdots P_s^{e_s}$ with $P_i$
  and $P_j$ coprime whenever $i\neq j$, $\deg P_i=d_i$ and $e_i\geq
  1$. From a computation point of view, one needs to introduce a
  variable $\alpha_{i,j}$ for each root of $P_i$ and then a variable
  $\beta_{i,j}$ for the square root of $\alpha_{i,j}$. This yields a
  total number of $2\sum_{i=1}^sd_i$ variables.  In
  Section~\ref{ss:sqrt_char_p}, we present another method which
  manages to introduce only $2s$ variables in characteristic $p>2$.
\subsection{Probabilistic and Deterministic
  Algorithms} \label{ss:proba_algo} We first describe a simple
probabilistic algorithm summarizing the method of
Section~\ref{ss:computation}.
\begin{algo}\label{al:one_sol} Probabilistic algorithm.
  \begin{description}
  \item[Input] Two sets of invertible symmetric matrices
    $\cH=\{H_1=D,\ldots,H_m\}\subseteq\K^{n\times n}$ and
    $\cH'=\{H_1'=D,\ldots,H_m'\}\subseteq\K^{n\times
      n}$.
  \item[Output] A description of the matrix $A \in \GL_n(\LL)$ such
    that $H_i'=A^{\T}\,H_i\,A$ for all $1\leq i\leq m$ or ``\NoSol''.
  \end{description}
  \begin{enumerate}
    \item Compute the vector subspace $\mcal{Y}=\{Y\ \mid
    \ D^{-1}\,H_i\,Y=Y\,D^{-1}\,H_i',\
    \forall\,1\leq i\leq m\}\subseteq\K^{n \times n}$.
    \item \textbf{If} $\mcal{Y}$ is reduced to the null matrix
    \textbf{then return} ``\NoSol''.
    \item Pick at random $Y\in\mcal{Y}$.
    \item Compute $Z=D^{-1}\,Y^{\T}\,D\,Y$ and 
    $J=T^{-1}Z\,T\in\LL^{n\times n}$, the Jordan normal
    form of $Z$ together with \(T\).
    \item Compute $G^{-1}$ the inverse of a square root of
    $J$.
    \item \textbf{Return} $Y\,T\,G^{-1}$ and $T$.
  \end{enumerate}
\end{algo}

\begin{theorem}\label{proba:algo}
  Algorithm~\ref{al:one_sol} is correct with probability at least
  $1-n/|\K|$ and runs in polynomial time.
\end{theorem}
\begin{proof}
  The correctness and the polynomial-time complexity of the algorithm
  come from Section~\ref{ss:computation}.
  After computing $\mcal{Y}$ and putting the equations defining its
  matrices in triangular form, one has to pick at random one matrix $Y
  \in \mcal{Y}$. By sampling the whole field $\K$ on these free
  variables, the probability that $\det Y=0$ is upper bounded by
  $n/|\K|$ thanks to Schwartz-Zippel-DeMillo-Lipton
  Lemma~(\cite{DL78,Z79}).
\end{proof}

\begin{remark}
  Let us recall that the conjugacy problem does not depend on the
  ground field (see~\cite{Pazzis2010}), \emph{i.e.}\ if there
  exists $Y\in\GL_n(\K')$, such that $H_i\,Y=Y\,H_i'$, then there exists
  $Y'\in\GL_n(\K)$ such that $H_iY'=Y'H_i'$. This allows us to extend
  $\K$ to a finite extension in order to decrease the probability of
  getting a singular matrix $Y$. Thus the success probability of
  Algorithm~\ref{al:one_sol} can be amplified to $1-n/|\K'|$ for any
  extension $\K'\supseteq\K$. The probability can be
  then made overwhelming large by considering extension of degree $O(n)$. In
  this case, the algorithm returns the description of a solution on
  $\K'(\sqrt{\zeta_1},\ldots,\sqrt{\zeta_r})$.
  Notice also that this algorithm can be turned into a
  deterministic algorithm using~\cite[Theorem~2]{DBLP:conf/issac/ChistovIK97}. 
  That is, there is
  a polynomial-time algorithm allowing to compute an invertible element in
  $\mcal{Y}$.  Furthermore, if one of the original Hessian matrices is
  already invertible, the computations of the essential variables of
  paragraph~\ref{par:variables} and the search of an equation with
  $n$ essential variables in
  paragraph~\ref{par:identity} can be done in a deterministic
  way. Whence, the whole algorithm is deterministic.
\end{remark}
The main Theorem~\ref{th:main} summarizes this remark together with
Theorem~\ref{proba:algo}.

\subsection{The binary Case}\label{ss:char_2}
In this section, we investigate fields of characteristic $2$.  Let
$\K=\F_q$ and $(\f,\g)\in \K[\x]^m \times \K[\x]^m$.
Instead of Hessian matrices, we consider equivalently upper triangular
matrices $H_1,\ldots,H_m$ and $H_1',\ldots,H_m'$ such that:
\[f_i(\x)=\x^{\T}\,H_i\,\x,\ g_i(\x)=\x^{\T}\,H_i'\,\x,\quad\forall\,1\leq
i\leq m.\]

For any matrix $M\in\K^{n\times n}$, let us denote
$\Delta(M)=\Diag(m_{11},\ldots,m_{nn})$ and $\Sigma(M)=M+M^{\T}$.
It is classical that if there exists \(A \in \GL_n(\K)\) such that
$\g(\x)=\f(A\cdot\x)$, then we also have
\begin{align}
  \Sigma(H_i')&=A^{\T}\,\Sigma(H_i)\,A,
  \label{eq:char2sigma}\\
  \Delta(H_i')&=\Delta(A^{\T}\,H_i\,A),\quad\forall\,i,\ 1\leq i\leq m.
  \label{eq:char2delta}
\end{align}
It suffices for this to expand $\f(A\cdot\x)$ and to consider the upper
triangular matrices. In a sense, $\Sigma(H_i)$ is the Hessian matrix
of $f_i$ and $\Delta(H_i)$ allow us to remember the $x_j^2$ terms
in $f_i$.
Combining two equations of~\eqref{eq:char2sigma} yields
$\Sigma(H_j')^{-1}\Sigma(H_i')=A^{-1}\Sigma(H_j)^{-1}\Sigma(H_i)A$
as long as $\Sigma(H_j)$ is invertible.
Let us notice that $\Sigma(H_i)$'s are symmetric matrices with
a zero diagonal, thus antisymmetric matrices with a zero diagonal.
We would like to stress out that in odd dimension, the determinant of a
symmetric matrix $S$ with a zero diagonal is always zero. Indeed, expanding
formula $\sum_{\sigma\in\mfk{S}_n}\prod_{i=1}^ns_{i,\sigma(i)}$ yields, for each
nonzero term $\prod_{i=1}^ns_{i,\sigma(i)}$, the term
$\prod_{i=1}^ns_{i,\sigma^{-1}(i)}=\prod_{i=1}^ns_{\sigma(i),i}=
\prod_{i=1}^ns_{i,\sigma(i)}$. Dimension $n$ being odd,
$\prod_{i=1}^ns_{i,\sigma(i)}$ cannot be the same term
as $\prod_{i=1}^ns_{i,\sigma^{-1}(i)}$. Hence they cancel each other.
One can also see
these matrices as projections of antisymmetric matrices over a ring of
characteristic $0$, namely $\Z_q$ the
unramified extension of the ring of dyadic integers of degree $\log_2 q$.
Let $\tilde{S}\in\Z_q^{n\times n}$ be antisymmetric such that
$\tilde{S}\mapsto S$.
Then $\det\tilde{S}=\det \tilde{S}^{\T}=\det(-\tilde{S})=(-1)^n\det\tilde{S}$,
hence $\det\tilde{S}=0$ and $\det S=0$.

Therefore, if $n$ is odd, then a linear combination of the
$\Sigma(H_i)$'s will always be singular. This can be related to
the \emph{irregular case} of the introduction.
\paragraph{Reduction to canonical representations in even dimension}
In this setting, we also rely on Assumption~\ref{as:regular}
to assume that a linear combination $\sum_{i=1}^m\lambda_i\,f_i$ is not
degenerate, and
$\lambda_1,\ldots,\lambda_m$ can be found
in randomized polynomial time.
Assuming $\lambda_1\neq 0$, we substitute the linear
combinations $\sum_{i=1}^m\lambda_i\,H_i$ and $\sum_{i=1}^m\lambda_i\,H_i'$
to $H_1$ and $H_1'$.

As a
consequence, we can find linear forms $\ell_1,\ldots,\ell_n$ in $\x$
such that, see~\cite[Theorem~6.30]{finite_fields}:
$f_1(\x)=\ell_1\ell_2+\ell_3\ell_4+\cdots+\ell_{n-1}\ell_n$
or
$f_1(\x)=\ell_1\ell_2+\ell_3\ell_4+\cdots+\ell_{n-1}\ell_n+\ell_{n-1}^2
+d\ell_n^2$, where $\tr_{\K}(d)=d+d^2+\cdots+d^{q/2}=1$.
After applying this change of variables,
$\Sigma(H_1)$ is always the following invertible block diagonal matrix:
\[\Sigma(H_1)=\Diag\left(\begin{pmatrix}0 &1\\1 &0\end{pmatrix},
  \ldots,\begin{pmatrix}0 &1\\1 &0\end{pmatrix}\right).\]
Following
paragraph~\ref{par:invertible}, we once again choose $\nu_i$ such that
$\Sigma(H_i+\nu_i\,H_1)$ is invertible and replace $H_i$ by $H_i+\nu_i\,H_1$.
Thus,
Proposition~\ref{prop:first_reduction} and
Theorem~\ref{th:second_reduction} become:
\begin{proposition}\label{prop:char_2}
  Let $n$ be an even integer and $\K$ be a field of characteristic $2$. Let 
  $(\f,\g)\in\K[x_1,\ldots,x_n]^m \times \K[x_1,\ldots,x_n]^m $ be regular
  quadratic homogeneous polynomials. There is a randomized polynomial-time
  algorithm which returns ``\NoSol'' only if $\f\not\sim\g$ or a new
  instance
  \[(\tilde{\f},\tilde{\g})=((\delta,\tilde{f}_2,\ldots,\tilde{f}_m),
  (\delta,\tilde{g}_2,\ldots,\tilde{g}_m))\K[x_1,\ldots,x_n]^m\times
  \K[x_1,\ldots,x_n]^m\] such that
  $\f\sim\g\iff\tilde{\f}\sim\tilde{\g}$.
  Furthermore, denoting $D$ the
  upper triangular matrix of
  $\tilde{f}_1=\tilde{g}_1=\delta$, \IPS comes down to a
  $\Sigma(D)$-Orthogonal
  Simultaneous Matrix Conjugacy problem, \emph{i.e.}\ conjugacy by an
  $\Sigma(D)$-orthogonal matrix under some constraints:
  \begin{align*}
    A^{\T}\,\Sigma(D)\,A&=\Sigma(D)\text{ and }\forall\,i,2\leq i\leq m,
    \Sigma(D)^{-1}\,\Sigma(H_i')=A^{-1}\,\Sigma(D)^{-1}\,\Sigma(H_i)\,A,\\
    \Delta(A^{\T}\,H_i\,A)&=\Delta(H_i').
  \end{align*}
\end{proposition}
\begin{proof}
  We mimic the proof of Proposition~\ref{prop:first_reduction}. We compute
  in randomized polynomial time $\lambda_1,\ldots,\lambda_m \in \K$
  such that $\vphi=\sum_{i=1}^{m}\lambda_i \cdot f_i$ is regular and
  we define $\gamma=\sum_{i=1}^{m}\lambda_i \cdot g_i$.  Assuming
  w.lo.g.\ $\lambda_1\neq 0$. We have then:
  \[\f \sim \g \iff (\vphi,f_2,\ldots,f_m) \sim (\gamma,g_2,\ldots,g_m).\]
  Computing $\delta$ the canonical quadratic form equivalent to $\vphi$
  yields a $P \in\GL_n(\K)$ such that
  $\tilde{\f}=(\tilde{\vphi}=\delta,\tilde{f_2},\ldots,
  \tilde{f_m})=(\vphi(P\x),f_2(P\x),\ldots,f_m(P\x))$.

  Then computing the canonical quadratic form
  equivalent to $\gamma$ allows us to compare it with $\delta$.
  If they are different, then $\f \not \sim \g$ and we return
  ``\NoSol''. Otherwise, the reduction is given by a matrix $Q \in
  \GL_n(\K)$ such that
  $\tilde{\g}=(\tilde{\gamma}=\delta,\tilde{g}_2,\ldots\tilde{g}_m)=
  (\gamma(Q\x),g_2(Q\x),\ldots,g_m(Q\x))$. Thus,
  $\tilde{\f}\sim\tilde{\g}$ if  and only if $\f\sim\g$.

  Finally, equations~\eqref{eq:char2sigma}
  $A^{\T}\,\Sigma(H_i)\,A=\Sigma(H_i')$ for all $i$, $1\leq i\leq m$ of
  can be rewritten as $A^{\T}\,\Sigma(D)\,A=\Sigma(D)$ and
  $\Sigma(D)^{-1}\,\Sigma(H_i')=A^{-1}\,\Sigma(D)^{-1}\,\Sigma(H_i)A$ for all $i$,
  $2\leq i\leq m$, while equations~\eqref{eq:char2delta}
  $\Delta(A^{\T}\,H_i\,A)=\Delta(H_i')$ for all $i$, $1\leq i\leq m$ remain.
\end{proof}
As a consequence, we designed the following algorithm to solve regular
instances of
quadratic-\IPS in even dimension over a field of characteristic $2$.
\begin{algo}\label{al:char_2} Probabilistic algorithm in
  characteristic $2$.
  \begin{description}
  \item[Input] Two sets of triangular matrices
    $\cH=\{H_1=D,\ldots,H_m\}\subseteq\K^{n\times n}$ and
    $\cH'=\{H_1'=D,\ldots,H_m'\}\subseteq\K^{n\times n}$ such that
    $H_1+H_1^{\T}$ and $H_1'+H_1'^{\T}$ are invertible.
  \item[Output] A description of the matrix $A \in \GL_n(\LL)$ such
    that $H_i'=A^{\T}\,H_i\,A$ for all $1\leq i\leq m$ or ``\NoSol''.
  \end{description}
  \begin{enumerate}
    \item Compute the vector subspace $\mcal{Y}=\{Y\ \mid
    \ \Sigma(D)^{-1}\,\Sigma(H_i)\,Y=Y\,\Sigma(D)^{-1}\,\Sigma(H_i'),\
    \forall\,1\leq i\leq m\}\subseteq\K^{n \times n}$.
    \item \textbf{If} $\mcal{Y}$ is reduced to the
    null matrix
    \textbf{then return} ``\NoSol''.
    \item Pick at random $Y\in\mcal{Y}$.
    \item Compute
    $Z=\Sigma(D)^{-1}\,\Sigma(Y)^{\T}\,\Sigma(D)\,\Sigma(Y)$ and 
    $J=T^{-1}Z\,T\in\LL^{n\times n}$, the Jordan normal
    form of $Z$ together with \(T\).
    \item \textbf{While} $J$ is not diagonal
      \begin{enumerate}
      \item Pick at random $Y\in\mcal{Y}$.
      \item Compute $Z$, $J$ and $T$ as above.
      \end{enumerate}
    \item Compute $G^{-1}$ the inverse of a square root of $J$.
    \item \textbf{Return} $Y\,T\,G^{-1}$ and $T$.
  \end{enumerate}
\end{algo}
The while loop comes from the fact that unlike other characteristics,
even if $Z$ is invertible, it might not have any square roots which
are polynomials in $Z$. In Section~\ref{ss:sqrt_char_2}, we prove that
there exists a square root of $Z$, which is a polynomial in $Z$ if and
only if $Z$ is diagonalizable.

\paragraph*{Open Question: The Irregular Case}
As stated above, in characteristic $2$, the irregular case seems to
cover more cases than in other characteristics. Indeed, what is called
usually a regular quadratic form in odd dimension falls in the
irregular case. However, from the Hessian matrix point
of view, an instance is irregular if all linear combinations of the Hessian
matrices are singular over the ground field. This allows us to unify
our statement about irregularity to all characteristics.

It seems to be an intriguing challenge to solve the binary case on
instances with regular quadratic forms, in particular in odd dimension.

\subsection{Benchmarks}\label{ss:bench}
We present in this section some timings of our algorithms over instances
of \IPS. We created instances $\mcal{H}=\{H_1,\ldots,H_m\}$ and
$\mcal{H}'=\{H_1',\ldots,H_m'\}$ which are randomly alternatively equivalent
over $\F_p$, equivalent over $\F_{p^2}$ but not $\F_p$ or not equivalent at all
over $\bar{\F}_p$, the algebraic closure of $\F_p$, for an odd $p$.
We report our timings in the following Table~\ref{tab:char_odd}
obtained using one core of
an \textsc{Intel Core i7} at $2.6 \si{\giga\hertz}$ running
\textsc{Magma~2.19}, \cite{magma}, on
\textsc{Linux} with $16\si{\giga}\mrm{B}$
of RAM. These timings corresponds to solving the linear system
which is the dominant part in our algorithm with complexity $O(n^{2\,\omega})$.
The code is
accessible on the first author's
webpage~\url{http://www-polsys.lip6.fr/~berthomieu/IP1S.html}.
To simplify the presentation,
we only considered the case when $m=n$. That is, we only considered $n$ matrices
of size $n$.

Since our matrices are randomly chosen, we apply the following
strategy.
We first solve the linear system
$H_1^{-1}\,H_i\,A=A\,H_1'^{-1}\,H_i'$, for all $i$, $2\leq i\leq
m$. In fact, in practice, $i=2,3$ give enough equations to retrieve
$A$ up to one
free parameter if $\mcal{H}$ and $\mcal{H}'$ are indeed equivalent. If
the matrices are not equivalent, this linear system will return the
zero matrix only.

Then, to determine $A$, we solve one quadratic equation amongst the ones given
by $A^{\T}\,H_1\,A=H_1'$. Let us notice that either all these equations can be
solved over $\F_p$ or none of them can. If they can, then $\mcal{H}$
and $\mcal{H}'$ are equivalent over
$\F_p$ and we have determined $A$ up to a sign,
otherwise $\mcal{H}$ and $\mcal{H}'$ are
only equivalent over $\F_{p^2}$ but not $\F_p$ and we also have computed such
an $A$. This yields Algorithm~\ref{al:crypto}.

\begin{algo}\label{al:crypto} Simplified  Algorithm.
  \begin{description}
  \item[Input] Two sets of generic invertible symmetric matrices
    $\cH=\{H_1,\ldots,H_m\}\subseteq\K^{n\times n}$ and
    $\cH'=\{H_1',\ldots,H_m'\}\subseteq\K^{n\times n}$.
  \item[Output] A matrix $A \in \GL_n(\K)$ such
    that $H_i'=A^{\T}\,H_i\,A$ for all $1\leq i\leq m$ or ``\NoSol''.
  \end{description}
  \begin{enumerate}
    \item Compute the vector subspace $\mcal{Y}=\{Y\ \mid
    \ H_1^{-1}\,H_i\,Y=Y\,{H_1'}^{-1}\,H_i',\
    \forall\,2\leq i\leq 3\}\subseteq\K^{n \times n}$.
    \item \textbf{If} $\mcal{Y}$ is reduced to a space of
    singular matrices \textbf{then return} ``\NoSol''.
    \item Determine $Y_0$ such that
      $\mcal{Y}=\acc{\lambda Y_0 \mid \lambda\in\K}$.
    \item Solve in $\lambda$ one equation
    $\lambda^2(Y_0^{\T}\,H_1\,Y_0)_{i,j}=(H_1')_{i,j}$ for
    a suitable pair $(i,j)$.
    \item Set $A=\lambda\,Y_0$.
    \item Pick at random $r\in\K^n$.
    \item Check that $A^{\T}\,H_i\,A\,r=H_i'\,r$ for all $i$,
    $1\leq i\leq m$.
    \item \textbf{Return} $A$.
  \end{enumerate}
\end{algo}

\paragraph{Complexity estimate} Taking the first $3$ matrix equations
$H_1^{-1}\,H_i\,Y=Y\,{H_1'}^{-1}H_i'$, in the $n^2$ unknowns, one
can solve this system in $O(n^{2\,\omega})$ operations in $\K$. Then,
one needs to determine
$\lambda$ by extracting one square root in $\K$ which can be done in
$O((\log q)^3)$ operations in $\K=\F_q$ with Tonelli--Shanks's
algorithm, \cite{Shanks1973}. Finally, one can check that $A^{\T}\,H_i\,A=H_i'$
for all $i$, with high probability, by picking up at random a vector
and checking 
that the products of this vector with both sets of matrices coincides. This
can be done in $O(m\,n^2)$ operations in $\K$.

Recall that, in practice, the best
matrix multiplication algorithm is due to~\cite{Strassen1969} whose
complexity is in $O(n^{\log_2 7})\subseteq O(n^{2.807})$. Thus, our complexity is
in $O(n^{5.615})$. This complexity is well
confirmed since multiplying by $2$ the sizes and the number of matrices
multiplies our timings roughly by at most $50$.

\begin{table}[t]
  \centering
  \begin{tabular}{|c|r|r|r|r|r|r|r|r|r|r|}
    \hline
    $n$     
    &$20$ &$30$ &$40$ &$50$ &$60$ &$70$ &$80$ &$90$ &$100$\\
    \hline
    Timings 
    &$0.040$ &$0.20$ &$0.84$ &$2.7$ &$7.5$ &$17$
    &$40$ &$79$ &$130$\\
    \hline
  \end{tabular}
  \caption{Timings for solving \IPS over $\F_{65521}$ in $\si{s}$.}
  \label{tab:char_odd}
\end{table}

In Table~\ref{tab:char_2}, we report our timings for solving the linear
system of our algorithm in
characteristic~$2$ presented in Section~\ref{ss:char_2}. Our method does not
differ much from the one in odd characteristic. We pick at random two sets of
$m$ upper triangular matrices over $\F_2$ which are either equivalent
over $\F_2$ or not equivalent at all over $\bar{\F}_2$, the algebraic closure of
$\F_2$. We first solve the linear system
$\Sigma(H_1)^{-1}\,\Sigma(H_i)\,A=A\,\Sigma(H_1')^{-1}\,\Sigma(H_i')$,
for all $i$, $2\leq i\leq m$. Let us notice that in dimension $2$, the
linear system does not
yield any information on $A$. In dimensions $8$ or more (resp. $4$ and $6$),
if $\cH$ and $\cH'$ are
equivalent, the linear system yields in general $A$ up to one free parameter
if $m\geq 3$ (resp. $m\geq 5$). Otherwise, it yields the zero matrix.
Then, it suffices to solve one of the quadratic equations amongst the
one given by
$A^{\T}\,\Sigma(H_1)\,A=\Sigma(H_1')$ and $\Delta(A^{\T}\,H_i\,A)=\Delta(H_i')$,
for all $i$, $1\leq i\leq m$  (see Proposition~\ref{prop:char_2}).

We compare the timings of both \textsc{MAGMA} and the C library
\textsc{M4RI}, due to~\cite{m4ri}.

Once again, our complexity in $O(n^{2\,\omega})$ is well confirmed by our timings.
Thanks to the linear system which totally determines $A$ up
to one free parameter, we just need to set this parameter to $1$ to obtain $A$.
This also explains why our timings are better than over $\F_{65521}$ although
it would seem a lot of quadratic equations must be solved.

\begin{table}[t]
  \centering
  \begin{tabular}{|c|r|r|r|r|r|r|r|r|r|r|}
    \hline
    $n$     &$20$ &$30$ &$40$ &$50$ &$60$ &$70$ &$80$ &$90$ &$100$\\
    \hline
    Timings (\textsc{MAGMA})
    &$0.010$ &$0.030$ &$0.080$ &$0.25\phantom{0}$ &$0.68$ &$1.4\phantom{0}$
    &$3.2\phantom{0}$ &$6.25$ &$16\phantom{.00}$\\
    \hline
    Timings (\textsc{M4RI})
    & & &$0.010$ &$0.030$ &$0.06$ &$0.14$
    &$0.27$ &$0.51$ &$0.91$\\
    \hline
  \end{tabular}
  \caption{Timings for solving \IPS over $\F_2$ in $\si{s}$.}
  \label{tab:char_2}
\end{table}

\section{Counting the Solutions: \#\IPS}
\label{s:enumeration}
In this part, we present a method for counting the number of solutions to
quadratic-\IPS. The main result is a consequence
of~\cite[Lemma~4.11]{Singla2010}. According to
Proposition~\ref{prop:homo}, this is
equivalent to enumerating all the invertible linear transformations on
the variables between two sets of quadratic homogeneous
polynomials. We provide here an upper bound on the number of
solutions.  We consider in this part regular quadratic homogeneous instances
$(\f,\g) \in \K[\x]^m \times \K[\x]^m$.

Let $\cH=\{H_1=D,\ldots,H_m\}$ and
$\cH^{\prime}=\{H^{\prime}_1=D,\ldots,H^{\prime}_m\}$ be the Hessian
matrices in $\K^{n\times n}$ of $\f$ and $\g$ respectively. Our
counting problem is equivalent to enumerating the number of
$D$-orthogonal matrices $X$ satisfying:
\begin{eqnarray}\label{eq:enum}
  X^{-1}\,D^{-1}\,H_i \, X=D^{-1}\,H^{\prime}_i,  \quad \forall\,i, 1 \leq i
  \leq m.
\end{eqnarray}
In~\cite[Section~4]{Singla2010}, the author computes the set of all
matrices commuting with a given matrix. In particular, from Lemma~4.11 of the
aforementioned paper, we can determine the size of this set and thus
our upper bound on the number of solutions to quadratic-\IPS. In order
to be self-contained, the proofs of the following lemmas shall be found
in~\ref{ap:enum_proof}.

Let us notice that if $X$ and $X'$ are both orthogonal solutions of
\eqref{eq:enum}, then $XX'^{-1}$ commutes with $D^{-1}\cH$ (resp.  $X^{-1}X'$
commutes with $D^{-1}\cH^{\prime}$). Therefore, the size of
the set of solutions is upper bounded by the number of invertible
elements in the centralizer $\Cent(D^{-1}\cH)$ of $D^{-1}\cH$.

Let $\alpha$ be an algebraic element of degree $m$ over $\K$ and let
$\K'=\K(\alpha)$.  We consider the matrix
$H=D^{-1}\pare{H_1+\cdots+\alpha^{m-1}H_m}\in\K'^{n\times n}$.  It is clear that a
matrix $X\in\K^{n\times n}$ is such that $X^{-1}\,D^{-1}\,H_i\,X=D^{-1}\,H_i$ for all
$i,1\leq i\leq m$ if and only if $X^{-1}HX=H$. Hence, the problem
again reduces itself to the computation of the centralizer $\Cent(H)$
of $H$ intersected with $\GL_n(\K)$.  To ease the analysis, we
consider the subspace $\mcal{V}=\Cent(H)\cap\K^{n\times n}$ of
matrices in $\K^{n\times n}$ commuting with $H$. This provides an
upper bound on the number of solutions.  The dimension of $\mcal{V}$
as a $\K$-vector space is upper bounded by the dimension of $\Cent(H)$
as a $\K'$-vector space. Indeed,
$\mcal{V}\otimes\K'\subseteq\Cent(H)$, hence
$\dim_{\K}\mcal{V}=\dim_{\K'}(\mcal{V}\otimes\K')\leq
\dim_{\K'}\Cent(H)$. Since we only want the size of the
centralizer of $H$, we can restrict our attention to the centralizer
of the Jordan normal form $J$ of $H$ defined over a field $\LL$.

Let us denote $\zeta_1,\ldots,\zeta_r$ the eigenvalues of $J$.
According to~\cite[Lemma~4.11]{Singla2010} and Lemma~\ref{lm:dim_cent}
in~\ref{ap:enum_proof}, if $J$ is made of Jordan
blocks of size $s_{i,1}\leq\cdots\leq s_{i,d_i}$ for $i,1\leq i\leq
r$, then centralizer of $H$ has dimension at most
\[\sum_{1\leq i\leq r}\sum_{1\leq j\leq d_i}(2d_i-2j+1)s_{i,j}.\]

As a consequence, if $q$ is an odd prime power, then the following
corollary gives an upper bound on the number of solutions of
quadratic-\IPS in $\F_q^{n\times n}$.
\begin{corollary}\label{cor:enum}
  Let $H_1,\ldots,H_m\in\F_q^{n\times n}$ be symmetric matrices. Let
  $\alpha$ be algebraic over $\F_q$ of degree $m$.
  Let $H=\sum_{i=1}^m\alpha^{i-1}\,H_i\in\F_{q^m}^{n\times n}$ and let
  $J$ be its normal Jordan form with eigenvalues
  $\zeta_1,\ldots,\zeta_r$. Assuming the blocks of $J$ associated
  to $\zeta_i$ are $J_{\zeta_i,s_{i,1}},\ldots,J_{\zeta_i,s_{i,d_1}}$ with
  $s_{i,1}\leq\cdots\leq s_{i,d_i}$ for $i,1\leq i\leq r$, then the
  number of solutions of quadratic-\IPS in $\F_q^{n\times n}$
  on the instance $(H_1,\ldots,H_m)$ is at most
  \[q^{\pare{\sum_{1\leq i\leq r}\sum_{1\leq j\leq d_i}(2\,d_i-2\,j+1)\,s_{i,j}}}-1.\]
\end{corollary}
As mentioned in the introduction,
the counting problem considered here is related to cryptographic
concerns. It corresponds to evaluating the number of equivalent
secret keys in \MPKC (see~\cite{DFPW11,WP11}). In particular,
in~\cite{DFPW11}, the authors propose an ``ad-hoc''
method for solving a particular
instance of \#\IPS. An interesting open question would be to revisit
the results from~\cite{DFPW11} with our approach.

\section{Special Case of the general IP Problem}\label{s:ip}
In this part, we present a randomized polynomial-time algorithm
for the following task:\\
\textbf{Input:} $\g=(g_1,\ldots,g_n) \in \K[x_1,\ldots,x_n]^n$,
and
$\POW_{n,d}=(x_1^d,\ldots,x_n^d) \in \K[x_1,\ldots,x_n]^n$ for some $d>0$.\\
\textbf{Question:} Find -- if any -- $(A,B) \in\GL_n(\K) \times
\GL_n(\K)$ such that:
\[\g=B \cdot \POW_{n,d}(A\cdot \x),
\mbox{ with $\x=(x_1,\ldots,x_n)^{\T}$.}\]
In~\cite{Kayal2011}, the author
proposes a randomized polynomial-time algorithm for solving the
problem when $B$ is the identity matrix and $m=1$ of finding $A$ such
that $g(\x)=f(A\cdot\x)$ with $f(\x)=\sum_{i=1}^nx_i^d$.

We generalize this result to $m=n$ with an additional transformation on
the polynomials. The main tool of our method is the following theorem.
\begin{theorem}\label{th:Kaltofen}
  Let  $\g=(g_1,\ldots,g_n)$ be polynomials of degree $d$ over
  $\K[x_1,\ldots,x_n]$ given in dense representation. 
  Let $L$ be a polynomial size of an
  arithmetic circuit to evaluate the determinant
  of the Jacobian matrix of $\g$. If the size
  of $\K$ is at least $12\,\max
  (2^{L+2},d\,(n-1)\,2^{d\,(n-1)}+d^3\,(n-1)^3,2\,(d\,(n-1)+1)^4)$,
  then one can factor the determinant of the Jacobian matrix of $\g$
  in randomized polynomial time.
\end{theorem}
\begin{proof}
  For this, we will use \cite{Kaltofen89factorizationof}'s algorithm
  to factor a polynomial given by evaluation, the needed size of $\K$
  is a consequence of this. As $\g$ has at most
  $n\,\binom{n+d-1}{d}\in O(n^{d+1})$ monomials, it can be evaluated in
  polynomial time using a multivariate Horner's scheme. Each
  $\frac{\partial g_i}{\partial x_j}(\mbf{a})$ is obtained as the
  coefficient in front of $x_j$ of the expansion of
  $g_i(a_1,\ldots,a_{j-1},a_j+x_j,a_{j+1},\ldots,a_n)$ which is a univariate
  polynomial of degree at most $d$. By~\cite[Chapter~1, Section~8]{BiniP1994},
  this can be computed as the shift of a polynomial in polynomial
  time. Hence the Jacobian matrix of $\g$ at $\mbf{a}$
  can be evaluated in polynomial time with an arithmetic circuit of
  polynomial size $L$.  This circuit can be for instance
  the evaluation of the Jacobian matrix of $\g$ followed by a Gaussian
  elimination on the matrix to compute the determinant. The
  determinant of the matrix can be recovered by linear algebra in $O(n^{\omega})$
  operations, with $\omega$ being the exponent of time complexity of
  matrix multiplication, $2\leq\omega\leq 3$. Using the arithmetic
  circuit of polynomial size $L$
  to evaluate the determinant of the Jacobian matrix, one can use
  \textsc{Kaltofen}'s algorithm to factor it in polynomial time.
\end{proof}
As in~\cite{Kayal2011} or~\cite{DBLP:conf/eurocrypt/Perret05},
we use partial derivatives to extract matrices $A$ and $B$. The
idea is to factor the Jacobian matrix
(whereas~\cite{Kayal2011} uses the Hessian matrix)
of $\g$ at $\x$ which is defined as follows:
\[\rJ_{\g}(\x) = \pare{\partial_j g_i=\frac{\partial g_i} {\partial
    x_j}}_{\substack{1\leq i\leq m\\1\leq j\leq n}}.\]
According to the following lemma, the Jacobian matrix is
especially useful in our context:
\begin{lemma}\label{lm:Jac}
  Let $\big( \f=(f_1,\ldots,f_m),\g=(g_1,\ldots,g_m)\big) \in
  \K[x_1,\ldots,x_n]^m \times \K[x_1,\ldots,x_n]^m$. If
  $(A,B) \in\GL_n(\K) \times \GL_m(\K)$ are such that $\g=B \cdot
  \f(A\cdot \x)$, then
  \[
  \rJ_{\g}(\x) = B \cdot \rJ_{\f}( A\cdot \x) \cdot A.
  \]
  As a consequence, $\det\rJ_{\g}(\x)=\det A \cdot \det B
  \cdot \det\rJ_{\f}( A\cdot \x)$.
\end{lemma}
As long as $\car\K$ does not divide $d$, the Jacobian matrix of
$\f=\POW_{n,d}(\x)$ is an invertible diagonal matrix whose
diagonal elements are $\big(\rJ_{\f}(\x) \big)_{i,i}= d \cdot
x_i^{d-1}$, $\forall\,i, 1 \leq i \leq n$. Thus:
\[
\det\rJ_{\POW_{n,d}}(\x)=d^n \prod_{i=1}^{n}x_i^{d-1}.
\]
This gives
\begin{lemma}\label{lm:det}
  Let $\g=(g_1,\ldots,g_n) \in \K[x_1,\ldots,x_n]^n$. Let $d>0$ be an
  integer, and define $\POW_{n,d}=(x_1^d,\ldots,x_n^d) \in
  \K[x_1,\ldots,x_n]^m$. If $(A,B) \in\GL_n(\K) \times
  \GL_n(\K)$ are such that $\g(\x)=B \cdot \POW_{n,d}(A\cdot \x)$, then:
  \[
  \det\rJ_{\g}(\x) = c \cdot \prod_{i=1}^{n} \ell_i(\x)^{d-1},
  \]
  with $c \in \K \setminus \{ 0\}$, and the $\ell_i$'s are
  linear forms whose coefficients are the $i$th rows of $A$.
\end{lemma}
\begin{proof}
  According to Lemma~\ref{lm:Jac}, $\det\big( \rJ_{\g}(\x)
  \big)=\det(A) \cdot \det(B) \cdot d^{n} \cdot \prod_{i=1}^{n}
  \ell_i(\x)^{d-1}$.
\end{proof}

From Lemmas~\ref{lm:det}, we can derive a randomized
polynomial-time algorithm
for solving \IP on the instance $(\f=\POW_{n,d},\g) \in \K[x_1,\ldots,x_n]^n
\times \K[x_1,\ldots,x_n]^n$ in characteristic $0$. It suffices to use
\cite{Kaltofen89factorizationof}'s
algorithm for factoring
$\det\rJ_{\g}(\x)$ in randomized polynomial time.

This allows us to recover -- if any -- the change of variables $A$.  The
matrix $B$ can be then recovered by linear algebra, \emph{i.e.}\
solving a linear system of equations. This proves the result announced
in the introduction for \IP, that is Theorem~\ref{th:ip} whenever
$\car\K\nmid d$.

\paragraph{Small characteristic}If $\car\K$ divides $d$, we must change
a little bit our strategy. Let us
write $d=p^re$ with $\car\K=p$ and $e$ coprime.
Then,
\begin{align*}
  \POW_{n,d}(A\,\x) &=
  \pare{\pare{\sum_{j=1}^na_{1,j}\,x_j}^{p^re},\ldots,
    \pare{\sum_{j=1}^na_{n,j}\,x_j}^{p^re}}\\
  \POW_{n,d}(A\,\x) &=\pare{\pare{\sum_{j=1}^na_{1,j}^{p^r}\,x_j^{p^r}}^e,\ldots,
    \pare{\sum_{j=1}^na_{n,j}^{p^r}\,x_j^{p^r}}^e}\\
  \POW_{n,d}(A\,\x) &=\POW_{n,e}\pare{A^{(p^r)}\,\x^{p^r}},
\end{align*}
with $A^{(p^r)}=\pare{a_{i,j}^{p^r}}_{1\leq i,j\leq n}$ and
$\x^{p^r}=\pare{x_1^{p^r},\ldots,x_n^{p^r}}$. Thus $\g$ is a polynomial in
$\x^{p^r}$ and by replacing $\x^{p^r}$ by $\x$, the problem comes down to
checking if $\tilde{\g}=B\cdot\POW_{n,e}(A^{(p^r)}\cdot\x)$ where
$\g(\x)=\tilde{\g}(\x^{p^r})$. Now, as
$\tilde{\g}(\x)=B\cdot\POW_{n,e}\pare{A^{(p^r)}\,\x}$, then
\[\rJ_{\tilde{\g}}=B\cdot\rJ_{\POW_{n,e}}(A^{(p^r)}\,\x)\cdot A^{(p^r)}.\]
Hence, $\det\rJ_{\tilde{\g}}(\x)=\det B\,\det\rJ_{\POW_{n,e}}(A^{(p^r)}\,\x)
\det A^{(p^r)}=e^n\prod_{i=1}^n
{\tilde{\ell}_i(\x)}^{e-1}(\det A)^{p^r}\det B$, where the
$\tilde{\ell}_i$'s are linear forms whose
coefficients are the $i$th rows of $A^{(p^r)}$. Then, to use \textsc{Kaltofen}'s
algorithm, one must set a low enough probability of failure $\eps$ 
yielding a big
enough set of sampling points, see~\cite[Section~6, Algorithm,
Step~R]{Kaltofen89factorizationof}. 
In particular, if the arithmetic circuit for evaluating the
determinant of the Jacobian we
want to factor has size $L$, then the size of the sampling set must be
greater than
\[\frac{6}{\eps}\max\pare{2^{L+2},e(n-1)\,2^{e(n-1)}+e^3\,(n-1)^3,
  2\,(e(n-1)+1)^4},\]
recalling that our polynomial has degree $e(n-1)$. 
In other words, if the probability of failure
is less than $1/2$, then one must consider a field of
size at least 
\[12\max\pare{2^{L+2},e(n-1)\,2^{e(n-1)}+e^3\,(n-1)^3,2\,(e(n-1)+1)^4}.
\]
All in all, this allows us to retrieve -- if any -- the change of
variables $A^{(p^r)}$ and thus $A$.
Then $B$ can be recovered by linear algebra. This proves Theorem~\ref{th:ip}
for any characteristic, as in the introduction.

\begin{reptheorem}{th:ip}
  Let $\g=(g_1,\ldots,g_n) \in \K[x_1,\ldots,x_n]^n$ be given in
  dense representation, and
  $\f=\POW_{n,d}=(x_1^d,\ldots,x_n^d) \in \K[x_1,\ldots,x_n]^n$ for
  some $d>0$. Whenever $\car\K=0$, let $e=d$ and
  $\tilde{\g}=\g$. Otherwise, let $p=\car\K$, let $e$ and $r$ be integers such
  that $d=p^r\,e$, $p$ and $e$ coprime, and let $\tilde{\g}\in\K[\x]^n$
  be such that $\g(\x)=\tilde{\g}(\x^{p^r})$. Let $L$ be a polynomial
  size of an arithmetic circuit to evaluate the determinant
  of the Jacobian matrix of $\tilde{\g}$. If the size
  of $\K$ is at least $12\,\max
  (2^{L+2},e(n-1)\,2^{e(n-1)}+e^3\,(n-1)^3,2\,(e(n-1)+1)^4)$,
  then there is a randomize
  polynomial-time algorithm which recovers -- if any --
  $(A,B)\in\GL_n(\K)\times\GL_n(\K)$ such that:
  \[\g=B\cdot\POW_{n,d}(A\cdot\x).\]
\end{reptheorem}
The computation of such a pair $(A,B)$ is summarized in the following Algorithm.
\begin{algo}\label{al:ip}\IP for $\POW_{n,d}$ and $\g$.
  \begin{description}
  \item[Input] One set of polynomials
    $\g=(g_1,\ldots,g_n)\in\K[x_1,\ldots,x_n]^n$, homogeneous of
    degree $d$.
  \item[Output] Two matrices $A,B\in\GL_n(\K)$ -- if any -- such that
    $\g(\x)=B\cdot\POW_{n,d}(A\cdot\x)$ with
    $\POW_{n,d}(\x)=(x_1^d,\ldots,x_n^d)$ or ``\NoSol''.
  \end{description}
  \begin{enumerate}
    \item \textbf{If} $\car\K=p>0$ \textbf{then}
      \begin{enumerate}
      \item Compute $r,e$ such that $d=p^r\,e$ with $p$
        and $e$ coprime.
      \item  Compute $\tilde{\g}$ such that $\tilde{\g}(\x^{p^r})=\g(\x)$.
      \end{enumerate}
    \item \textbf{Else} $e=d$ and $\tilde{\g}=\g$.
    \item Create an arithmetic circuit of polynomial size $L$
    to evaluate $\det\rJ_{\g}(\x)$.
    \item Evaluate $\det\rJ_{\g}(\x)$ in at least $12\,\max
    (2^{L+2},e(n-1)\,2^{e(n-1)}+e^3\,(n-1)^3,2\,(e(n-1)+1)^4)$
    distinct points.
    \item Factor $\det\rJ_{\g}(\x)$ with \textsc{Kaltofen}'s algorithm.
    \item \textbf{If} the factorization is
    $c\,\prod_{i=1}^n\ell_i(\x)^{e-1}$ \textbf{then}
    $A=\pare{\ell_{i,j}^{e/d}}_{1\leq i,j\leq n}$.
    \item \textbf{Else return} ``\NoSol''.
    \item Compute $B$ such that $\tilde{\g}(\x)=B\cdot
    \pare{\ell_1(\x)^e,\ldots,\ell_n(\x)^e}^{\T}$.
    \item \textbf{Return} $(A,B)$.
  \end{enumerate}
  
\end{algo}

\section{Square Root of a Matrix}\label{s:matsqrt}
In this section, we present further algorithms for computing the
square root of a matrix. We use the same notation as in
Section~\ref{s:IP1S}. A square root of a matrix $Z$ is a
matrix whose square is $Z$. In the first subsection, we deal with some
properties of the square root of a matrix in characteristic not
$2$. In particular, we show that an invertible matrix $Z$ always has a
square root which is a polynomial in $Z$. In the second subsection,
we consider the case of characteristic $2$. We recall that
whenever $Z$ is not diagonalizable, then $Z$ might have a square root
but it is never a polynomial in $Z$. We give some examples of such
matrices $Z$.  Lastly, we propose an alternative to the method of
Section~\ref{s:IP1S} for computing the square root of a matrix in
polynomial time for any field of characteristic $p\geq2$.

\subsection{The square root as a polynomial in characteristic not
  $2$}\label{ss:sqrt}
In this part, we prove that an invertible matrix always has a square
root which is a polynomial in considered matrix. More specifically, we
shall prove the following result.
\begin{proposition}
  Let $Z\in\K^{n\times n}$ be an invertible matrix whose eigenvalues are
  $\zeta_1,\ldots,\zeta_r$.  Let $\omega_1,\ldots,\omega_r$ be such that
  $\omega_i^2=\zeta_i$ and $\zeta_i=\zeta_j \Rightarrow
  \omega_i=\omega_j$, for all $1\leq i,j\leq r$. Then, there exists
  $W\in\K(\omega_1,\ldots,\omega_r)[Z]$ a square root of $Z$
  whose eigenvalues are $\omega_1,\ldots,\omega_r$.
\end{proposition}
\begin{proof}
  Let $T$ be a matrix of change of basis, such that $J=T^{-1}\,Z\,T$ is made
  of Jordan blocks. It is clear that $W$ such that $W^2=Z$ is a
  polynomial in $Z$, \emph{i.e.}\ $W=Q(Z)$, if and only if
  $G=T^{-1}\,W\,T$ satisfies
  $G=Q(J)$. Let $J_{\zeta,d}$ be the Jordan block of size $d$ associated with
  eigenvalue $\zeta$ and $\omega$ be a square root of $\zeta$. We
  shall first prove that the square root $G_{\omega,d}$ of $J_{\zeta,d}$
  is a polynomial in $J_{\zeta,d}$ with coefficients in $\K(\omega)$.
  Matrix $J_{\zeta,d}-\zeta\,\id_d$ is nilpotent of degree $d$. Hence, by the
  classical Taylor expansion of the square root near $\id_d$, one can write
  \begin{align}
    G_{\omega,d}&=\omega\sum_{k=0}^{d-1}
    \binom{1/2}{k}\,\zeta^{-k}\pare{J_{\zeta,d}-\zeta\,\id}^k
    =\sum_{k=0}^{d-1}\binom{1/2}{k}\,\omega^{1-2\,k}
    \pare{J_{\zeta,d}-\zeta\,\id}^k
    =Q_{\zeta}(J_{\zeta,d})\label{eq:sqrtmat}\\
    &=\begin{pmatrix}
      \omega  &\binom{1/2}{1}\,\omega^{-1} &\cdots
      &\binom{1/2}{d-1}\,\omega^{3-2\,d}\\
      &\ddots &\ddots &\vdots\\
      &&\ddots &\binom{1/2}{1}\,\omega^{-1}\\
      &&&\omega
    \end{pmatrix},\notag
  \end{align}
  with $Q_{\zeta}(x)=\sum_{k=0}^{d-1}\binom{1/2}{k}\,\omega^{1-2\,k}
  \pare{x-\zeta}^i\in\K(\omega)[x]$.

  It remains to prove that for multiple Jordan blocks, one can find a
  common polynomial.
  From equation~\eqref{eq:sqrtmat}, we deduce that $G$ is a polynomial
  in $J=\Diag\pare{J_{\zeta_1,d_1},\ldots,J_{\zeta_r,d_r}}$ if and
  only if there exists a polynomial $Q$ such that $Q=Q_{\zeta_i} \mod
  (X-\zeta_i)^{d_i}$, for all $i$, $1\leq i\leq r$. By the Chinese
  Remainder Theorem, this can always be solved as soon as
  $\zeta_i=\zeta_j$ implies $Q_{\zeta_i}=Q_{\zeta_j}\bmod
  (X-\zeta_i)^{\min(d_i,d_j)}$, which is exactly the condition
  $\omega_i=\omega_j$.
\end{proof}
Let us notice that picking the same square root for two equal
eigenvalues is necessary. Indeed, although
$W=\left(\begin{smallmatrix}1 &\phantom{-}0\\0
    &-1\end{smallmatrix}\right)$ is a square root
of $Z=\left(\begin{smallmatrix}1 &0\\0 &1\end{smallmatrix}\right)$,
$W\not\in\K[Z]$.

\subsection{Matrices with square roots in characteristic~$2$}
\label{ss:sqrt_char_2}
In this part, we consider the trickier case of computing the square
root of a matrix over a field $\K$ with $\car\K=2$. Unfortunately,
unlike other characteristics, an invertible matrix has not necessarily a
square root over $\bar{\K}$. In fact, no Jordan block of size at
least $2$ has any square root. This is mainly coming from
the fact that generalized binomial coefficients $\binom{1/2}{k}$, involved
in the Taylor expansion, are meaningless in characteristic $2$.
\begin{proposition}
  Let $Z\in\K^{n\times n}$ be a Jordan normal form with blocks
  $J_1,\ldots,J_r$ of sizes $d_1,\ldots,d_r\geq 2$, associated to
  eigenvalues $\zeta_1,\ldots,\zeta_r$ and blocks of sizes $1$ with
  eigenvalues $\upsilon_1,\ldots,\upsilon_s$.  We assume that
  $J_1,\ldots,J_r$ are ordered by decreasing sizes and then
  eigenvalues. Matrix $Z$ has a square root $W$ if and only if
  $d_1-d_2\leq 1$ and $\zeta_1=\zeta_2$, $d_3-d_4\leq 1$ and
  $\zeta_3=\zeta_4$, \emph{etc.}\ and if for each $J_i$ of size $2$
  that is not paired with $J_{i-1}$ or $J_{i+1}$, then there exists a
  $j$ such that $\upsilon_j=\zeta_i$.

  Furthermore, matrix $W$ is a polynomial in $Z$ if and only if $Z$
  is diagonalizable.
\end{proposition}
Before, proving this result, we give some example of matrices with or
without square roots. Following matrices $J$ and $J'$ both have two
Jordan blocks associated with eigenvalue $\zeta$. Denoting $\omega$
the square root of $\zeta$, then $K$ is the square root of $J$
and $K_1',K_2'$ are those of $J'$ for $x,y,z$ any.
\begin{align*}
  J&=\begin{pmatrix}\zeta&0&0\\0&\zeta&1\\0&0&\zeta\end{pmatrix},\quad
  K=\begin{pmatrix}\omega&0&x\\\frac{1}{x}&\omega&y\\0&0&\omega
  \end{pmatrix},\\
  J'&=\begin{pmatrix}\zeta&1&0&0\\0&\zeta&0&0\\
    0&0&\zeta&1\\0&0&0&\zeta\end{pmatrix},\quad
  K_1'=\begin{pmatrix}\omega&x&0&y\\0&\omega&0&0\\
    \frac{1}{y}&z&\omega&x\\0&\frac{1}{y}&0&\omega\end{pmatrix},\quad
  K_2'=\begin{pmatrix}\omega&x&y&z\\0&\omega&0&y\\
    0&\frac{1}{y}&\omega&x\\0&0&0&\omega\end{pmatrix}.
\end{align*}
As one can see, none of $K$, $K_1'$ and $K_2'$ are polynomials in $J$
or $J'$ because of the nonzero subdiagonal elements $1/x$ and
$1/y$. Examples of matrices without square roots are $J''$, with two Jordan
blocks associated with $\zeta$ of sizes $1$ and $3$, and $J'''$,
with three Jordan blocks associated with $\zeta$ of size $2$. Computing
a square root of each of them yields an inconsistent system.
\begin{align*}
  J''&=\begin{pmatrix}\zeta&0&0&0\\0&\zeta&1&0\\
    0&0&\zeta&1\\0&0&0&\zeta\end{pmatrix},\quad
  J'''=\Diag\pare{\begin{pmatrix}\zeta&1\\0&\zeta\end{pmatrix},
    \begin{pmatrix}\zeta&1\\0&\zeta\end{pmatrix},
    \begin{pmatrix}\zeta&1\\0&\zeta\end{pmatrix}}.
\end{align*}

\begin{proof}
  Let $J$ be a Jordan block of size $d$ associated to eigenvalue $\zeta$.
  Then $J^2-\zeta^2\id=\pare{\begin{smallmatrix}\mbf{0}&\id_{d-2}\\
      \mbf{0}&\mbf{0}\end{smallmatrix}}$ and one can deduce that
  $\zeta^2$ is the sole
  eigenvalue of $J^2$ but that its geometric multiplicity is
  $2$. Hence the Jordan normal form of $J^2$ is made of two Jordan
  blocks.

  As $(J-\zeta\id)^d=0$ and $(J-\zeta\id)^e\neq 0$ for all $e<d$, then
  $\pare{J^2-\zeta^2\id}^{\ceil{d/2}}=0$ and
  $\pare{J^2-\zeta^2\id}^e\neq 0$ for $e<\ceil{d/2}$, \emph{i.e.}\
  $e<d/2$ if $d$ is even and $e<(d+1)/2$ if $d$ is odd. Thus the
  Jordan normal form of $J^2$ has a block of size exactly
  $\ceil{d/2}$. That is, if $d$ is even, both blocks have size $d/2$
  and if $d$ is odd, one block has size $(d+1)/2$ and the other block
  has size $(d-1)/2$.

  By this result, if $Z$ is a square, then one must be able to pair up
  its Jordan blocks with same eigenvalue $\zeta$ so that the sizes
  differ by at most $1$. The blocks that need not be paired being
  the blocks of size $1$.

  Conversely, assuming one can pair up the Jordan blocks of $Z$ with
  same eigenvalue $\zeta$ so that the sizes differ by at most $1$ and
  the remaining blocks have sizes $1$. Then, each pair of blocks is
  the Jordan normal form of the square of a Jordan block of size the
  sum of the sizes and eigenvalue $\sqrt{\zeta}$. Furthermore, each
  lonely block of size $1$ associated with $\zeta$ is the square of the
  block of size $1$ associated with $\sqrt{\zeta}$.

  Finally, for the last statement, the if part is easy. It remains the
  only if part for which we assume $W^2=Z$ and $Z$ is not diagonalizable.
  Let $J$ be the Jordan normal
  form of $Z$ with blocks $J_1,\ldots,J_r$. For any polynomial $P$, $P(J)$ is
  also block diagonal with blocks $P(J_1),\ldots,P(J_r)$. Thus, if
  $P(J)^2=J$, then $P(J_i)^2=J_i$ for all $1\leq i\leq r$, which is
  false, unless $J_i$ has size $1$.
\end{proof}
\subsection{Computation in
  characteristic~$p\geq2$}\label{ss:sqrt_char_p}
In this part, we present an alternative method to the one presented in
Section~\ref{ss:computation}. We aim at diminishing the number of
variables needed in the expression of the square root. However, this
method does not work in characteristic $0$. For the time being, we
consider $\car\K>2$.  However, we shall see below how to adapt this
method to the characteristic $2$.

The idea is still to perform a change of basis $T$ over $\K$ so that
$J=T^{-1}\,Z\,T$ has an easily computable square root. This matrix $J$ is
the \emph{generalized Jordan normal form}, also known as the
\emph{primary rational canonical form} of $Z$. As the classical Jordan
normal form, if $Z$ is diagonalizable over $\bar{\K}$, then $J$ is
block diagonal, otherwise it is a block upper triangular matrix. Its
diagonal blocks are companion matrices $\cC(P_1),\ldots, \cC(P_r)$ of
irreducible factors $P_1,\ldots,P_r$ of its characteristic polynomial.
Superdiagonal blocks are zero
matrices with eventually a $1$ on the bottom-left corner, if the
geometric multiplicity associated to the roots of one the $P_i$ is not
large enough. In other words, it gathers $d$ conjugated eigenvalues in
one block of size $d$ which is the companion matrix of their shared
minimal polynomial.  Let us note that computing such a normal form
can be done in polynomial time and
that the change of basis matrix $T$ is defined over $\K$,
see~\cite{Matthews1992,Storjohann1998}. Thus, after
computing a square root $G$ of $J$, one can retrieve $W$ and $A$ of
Section~\ref{ss:computation} 
in $O(n^{\omega})$ operations in the field of coefficients of
$G$, with $\omega$ being the exponent of the time-complexity of matrix
multiplication $2\leq\omega\leq 3$. Furthermore, computing a square
root of $J$ is equivalent to
computing the square root of each companion matrix. Finally, using the
same argument as for the more classical Jordan normal form in
Section~\ref{ss:sqrt}, $G$ is a polynomial in $J$. In the following,
we only show how to determine the square root of a companion matrix
$\cC(P)$, for an irreducible $P$.

Let $P=x^d+p_{d-1}\,x^{d-1}+\cdots+p_{0}$, let us recall that the
companion matrix of $P$ is
\[\cC(P)=\begin{pmatrix}0 & & &-p_0\\1 &\ddots & &-p_1\\
  &\ddots &0&\vdots\\ & &1 &-p_{d-1}\end{pmatrix}.\]
If polynomial $P$ can be decomposed as
$P(z)=(z-\alpha_0)\,\cdots\,(z-\alpha_{d-1})$, then we want to find a
polynomial $Q$ such that $Q(z)=(z-\beta_0)\,\cdots\,(z-\beta_{d-1})$,
where $\beta_i^2=\alpha_i$ for all $0\leq i\leq d-1$. Let us notice
that
\[P(z^2)=(z^2-\alpha_0)\,\cdots\,(z^2-\alpha_{d-1})
=Q(z)\,(z+\beta_0)\,\cdots\,(z+\beta_{d-1})
=(-1)^d\,Q(z)\,Q(-z).\]
As a consequence, the characteristic polynomial of
$\cC(Q)^2$ is
\[\det(\lambda\id-\cC(Q)^2)=\det(\sqrt{\lambda}\id-\cC(Q))
\det(\sqrt{\lambda}\id+\cC(Q))=(-1)^dQ(\sqrt{\lambda})Q(-\sqrt{\lambda})
=P(\lambda).\]
But since $P$
is irreducible over $\K$, by the invariant factors theory, then
$\cC(Q)^2$ must be similar to the companion matrix $\cC(P)$ over $\K$.

As $P$ is irreducible over $\K=\F_q$, up to reindexing the roots of
$P$, the conjugates $\alpha_1,\ldots,\alpha_{d-1}$ of $\alpha_0$ are
just its iterated $q$th powers. Denoting $\LL=\K[x]/(P(x))=\F_{q^d}$,
let us assume that $S(y)=y^2-x$ is reducible in $\LL[y]$, then
$\beta_0\in\LL$. As such, one can choose $\beta_i=\beta_0^{q^i}$, the
iterated $q$th powers. In that case, the previous equations can be
rewritten
\begin{align*}
  P(z)&=(z-\alpha_0)\,\pare{z-\alpha_0^q}\,\cdots\,\pare{z-\alpha_0^{q^{d-1}}}
  =(z-x)\,\pare{z-x^q}\,\cdots\,\pare{z-x^{q^{d-1}}},\\
  Q(z)&=(z-\beta_0)\,\pare{z-\beta_0^q}\,\cdots\,\pare{z-\beta_0^{q^{d-1}}}
  =(z-y)\,\pare{z-y^q}\,\cdots\,\pare{z-y^{q^{d-1}}}.
\end{align*}
As a consequence, $Q(z)\in\K[z]$ and to compute $Q(z)$, we need to
compute $y^{q^i}$ effectively. This is done by computing the following
values in $O(d\log q)$ operations in $\LL$:
\[u_0=x, u_1=x^q\bmod P(x),\ldots,u_{d-1}=u_{d-2}^q=x^{q^{d-1}}\bmod
P(x).\] Then, we simply compute in $d$ operations
$Q(z)=(z-u_0)\,(z-u_1)\,\cdots\,(z-u_{d-1})$ and we know that the resulting
polynomial is in $\K[z]$.

Whenever $\alpha_0$ is not a square in $\LL$, that is whenever $S(y)$
is irreducible, then $\beta_0^{q^d}$ is a square root of $\alpha_0$
different from $\beta_0$, thus it is $-\beta_0$. As a consequence,
setting $Q(z)=(z-\beta_0)\,(z-\beta_0^q)\,\cdots\,(z-\beta_0^{q^{d-1}})$
would yield a polynomial that is not stable by the Frobenius
endomorphism.

As such, we introduce a new variable $y$ to represent the field
$\LL'=\LL[y]/(y^2-x)$ and to compute $Q(z)$, we need to compute
$y^{q^{i(d+1)}}$ effectively.  Since
$y^{q^i}=y\,y^{q^{i}-1}=y\,x^{\frac{q^{i}-1}{2}}$, we can compute the
following values in $O(d\log q)$ field operations in $\LL$:
\[u_0 = 1, u_1 = x^{\frac{q-1}{2}}\bmod P(x),\ldots,
u_{d-1}=u_1\,u_{d-2}^q=x^{\frac{q^{d-1}-1}{2}}\bmod P(x).\]
Consequently, $Q(z)=(z-yu_0)\,(z-yu_1)\,\cdots\,(z-yu_{d-1})$.

As a first step, we compute in $d$ operations, the dehomogenized
polynomial in $y$,
\[\tilde{Q}(z)=(z-u_{0}) (z-u_{1})\cdots (z- u_{d-1})
=z^d +h_{1} z^{d-1} + \cdots + h_{d-1} z + h_{d}.\] Then, $Q(z)=z^d +y
h_{1} z^{d-1} + \cdots + y^{d-1}h_{d-1} z + y^{d}h_{d}$.  But,
denoting by $i_0=i\bmod 2$, we have $y^i=y^{i_0}
y^{i-i_0}=y^{i_0}x^{\frac{i-i_0}{2}}$. Hence we deduce:
\begin{align*}
  Q(z)&=z^d +y h_{1} z^{d-1} +xh_{2}z^{d-2}+yxh_{3}z^{d-3}+ \cdots
  +y^{d_0}x^{\frac{d-d_0}{2}}h_d\\
  &=z^d+y\,\sum_{i=0}^{\floor{\frac{d-1}{2}}} h_{2\,i+1}\,
  x^{i}\,z^{d-2\,i-1} +\sum_{i=1}^{\floor{\frac{d}{2}}}h_{2\,i}\,x^{i}\,
  z^{d-2\,i}.
\end{align*}

\paragraph{Complexity analysis}
Since the number of operations for computing the square root of a
block of size $d$ is bounded by $O(d\,\log q)$ operations in
$\LL=\F_{q^d}$, this is also bounded by $O(d\,\msf{M}(d)\,\log q)$
operations in $\K=\F_q$, where
$\msf{M}(n)$ is a bound on the number of operations in $\K$ to
multiply two polynomials in $\K[x]$ of degree at most $n-1$.
As a consequence, the computation of $W$ can
be done in no more than $O(n^{\omega}+n\msf{M}(n)\log q)$ operations
in $\K$. Let us assume that the characteristic polynomial of $Z$ has
degree $n$ and can be factored as $P_1^{e_1}\cdots P_s^{e_s}$ with
$P_i$ and $P_j$ coprime whenever $i\neq j$, $\deg P_i=d_i$ and
$e_i\geq 1$. From a computation point of view, in the worst case, one
needs to introduce a variable $\alpha_i$ for one root of $P_i$ and a
variable $\beta_i$ for the square root of $\alpha_i$, assuming
$\alpha_i$ is not a square.  This yields a total number of $2s$
variables.

\paragraph{Computation in characteristic $2$}
The case of characteristic $2$ is almost the same. From a polynomial
$P(z)=z^d+p_{d-1}z^{d-1}+\cdots+p_0=(z-\zeta_1)\cdots(z-\zeta_d)$, we
want to compute
$Q(z)=z^d+q_{d-1}z^{d-1}+\cdots+q_0=(z-\omega_1)\cdots(z-\omega_d)$,
with $\omega_i^2=\zeta_i$ for all $1\leq i\leq d$. As $P(z^2)=Q(z)^2$,
this yields $q_i=\sqrt{p_i}=p_i^{q/2}$, for all $1\leq i\leq
d-1$. Thus, $Q$ can be computed in $O(d\log q)$ operations in $\K$ and
as a consequence, $W$ in $O(n^{\omega}+n\log q)$ operations in $\K$.

However, let us recall that $D$ is block diagonal if and only if the
Jordan normal form is block diagonal. As such, a square root of $D$ is
a polynomial in $D$ if and only if $D$ is block diagonal, see
Section~\ref{ss:sqrt_char_2}.

\subsection*{Acknowledgements}
We would like to thank Gabor \textsc{Ivanyos} for his helpful remarks
and references on
the irregular case. We wish to thank Gilles \textsc{Macario-Rat} for
the many discussions about isomorphism of quadratic polynomials and Nitin
\textsc{Saxena} for those about graph isomorphism.

We thank the anonymous referees for their careful reading and their helpful
comments.

This work has been partly supported by
the French National Research Agency \textsc{ANR-11-BS02-0013 HPAC} project.

\bibliographystyle{elsarticle-harv} 
\bibliography{main}

\begin{thebibliography}{59}
\expandafter\ifx\csname natexlab\endcsname\relax\def\natexlab#1{#1}\fi
\expandafter\ifx\csname url\endcsname\relax
  \def\url#1{\texttt{#1}}\fi
\expandafter\ifx\csname urlprefix\endcsname\relax\def\urlprefix{URL }\fi

\bibitem[{Agrawal and Saxena(2006)}]{DBLP:conf/stacs/AgrawalS06}
Agrawal, M., Saxena, N., 2006. Equivalence of {F}-{A}lgebras and {C}ubic
  {F}orms. In: Durand, B., Thomas, W. (Eds.), STACS. Vol. 3884 of Lecture Notes
  in Computer Science. Springer, pp. 115--126.

\bibitem[{Albrecht and Bard(2012)}]{m4ri}
Albrecht, M., Bard, G., 2012. {The M4RI Library -- Version 20121224}. The
  M4RI~Team.
\newline\urlprefix\url{http://m4ri.sagemath.org}

\bibitem[{Bernardi et~al.(2011)Bernardi, , and
  Gimigliano}]{DBLP:journals/jsc/BernardiGI11}
Bernardi, A., , Gimigliano, A.~Id{\`a}, M., 2011. Computing symmetric rank for
  symmetric tensors. J. Symb. Comput. 46~(1), 34--53.

\bibitem[{Berthomieu et~al.(2010)Berthomieu, Hivert, and
  Mourtada}]{BerthomieuHM2010}
Berthomieu, J., Hivert, P., Mourtada, H., 2010. Computing {H}ironaka's
  invariants: {R}idge and {D}irectrix. In: Arithmetic, Geometry, Cryptography
  and Coding Theory 2009. Vol. 521 of Contemp. Math. Amer. Math. Soc.,
  Providence, RI, pp. 9--20.

\bibitem[{Bettale et~al.(2013)Bettale, Faug\`ere, and Perret}]{BettaleFP2013}
Bettale, L., Faug\`ere, J.-C., Perret, L., 2013. Cryptanalysis of {HFE},
  {Multi-HFE} and {V}ariants for {O}dd and {E}ven {C}haracteristic. Designs,
  Codes and Cryptography 69~(1), 1 -- 52.

\bibitem[{Bhattacharyya et~al.(2013)Bhattacharyya, Fischer, and
  Lovett}]{DBLP:conf/soda/BhattacharyyaFL13}
Bhattacharyya, A., Fischer, E., Lovett, S., 2013. Testing low complexity
  affine-invariant properties. In: Proceedings of the Twenty-Fourth Annual
  ACM-SIAM Symposium on Discrete Algorithms. pp. 1337--1355.

\bibitem[{Bini and Pan(1994)}]{BiniP1994}
Bini, D., Pan, V.~Y., 1994. Polynomial and {M}atrix {C}omputations. {V}olume 1:
  Fundamental Algorithms. Progress in Theoretical Computer Science.
  Birkh\"auser Boston Inc., Boston, MA.

\bibitem[{Bosma et~al.(1997)Bosma, Cannon, and Playoust}]{magma}
Bosma, W., Cannon, J., Playoust, C., 1997. The {M}agma algebra system. {I}.
  {T}he user language. J. Symbolic Comput. 24~(3-4), 235--265, computational
  algebra and number theory (London, 1993).

\bibitem[{Bouillaguet et~al.(2011)Bouillaguet, Faug{\`e}re, Fouque, and
  Perret}]{DBLP:conf/pkc/BouillaguetFFP11}
Bouillaguet, C., Faug{\`e}re, J.-C., Fouque, P.-A., Perret, L., 2011. Practical
  cryptanalysis of the identification scheme based on the isomorphism of
  polynomial with one secret problem. In: Catalano, D., Fazio, N., Gennaro, R.,
  Nicolosi, A. (Eds.), Public Key Cryptography. Vol. 6571 of Lecture Notes in
  Computer Science. Springer, pp. 473--493.

\bibitem[{Bouillaguet et~al.(2013)Bouillaguet, Fouque, and
  V{\'e}ber}]{DBLP:conf/eurocrypt/BouillaguetFV13}
Bouillaguet, C., Fouque, P.-A., V{\'e}ber, A., 2013. Graph-theoretic algorithms
  for the "isomorphism of polynomials" problem. In: Johansson, T., Nguyen,
  P.~N. (Eds.), EUROCRYPT. Vol. 7881 of Lecture Notes in Computer Science.
  Springer, pp. 211--227.

\bibitem[{B{\"u}rgisser(2012)}]{DBLP:conf/coco/Burgisser12}
B{\"u}rgisser, P., 2012. Prospects for geometric complexity theory. In: IEEE
  Conference on Computational Complexity. IEEE, p. 235.

\bibitem[{B{\"u}rgisser and Ikenmeyer(2011)}]{DBLP:conf/stoc/BurgisserI11}
B{\"u}rgisser, P., Ikenmeyer, C., 2011. Geometric complexity theory and tensor
  rank. In: Fortnow, L., Vadhan, S.~P. (Eds.), STOC. ACM, pp. 509--518.

\bibitem[{B{\"u}rgisser and Ikenmeyer(2013)}]{DBLP:conf/stoc/BurgisserI13}
B{\"u}rgisser, P., Ikenmeyer, C., 2013. Explicit lower bounds via geometric
  complexity theory. In: Boneh, D., Roughgarden, T., Feigenbaum, J. (Eds.),
  STOC. ACM, pp. 141--150.

\bibitem[{Cai(1994)}]{Cai1994}
Cai, J., 1994. Computing {J}ordan {N}ormal forms {E}xactly for {C}ommuting
  {M}atrices in {P}olynomial {T}ime. International Journal of Foundations of
  Computer Science 05~(03n04), 293--302.

\bibitem[{Carlini(2005)}]{Carlini05reducingthe}
Carlini, E., 2005. Reducing the number of variables of a polynomial. In:
  Algebraic geometry and geometric modeling. Springer, pp. 237--247.

\bibitem[{Carlitz(1954)}]{Carlitz1954}
Carlitz, L., 03 1954. Representations by quadratic forms in a finite field.
  Duke Mathematical Journal 21~(1), 123--137.

\bibitem[{Chen et~al.(2011)Chen, Kayal, and
  Wigderson}]{DBLP:journals/fttcs/ChenKW11}
Chen, X., Kayal, N., Wigderson, A., 2011. Partial derivatives in arithmetic
  complexity and beyond. Foundations and Trends in Theoretical Computer Science
  6~(1-2), 1--138.

\bibitem[{Chistov et~al.(1997)Chistov, Ivanyos, and
  Karpinski}]{DBLP:conf/issac/ChistovIK97}
Chistov, A.~L., Ivanyos, G., Karpinski, M., 1997. Polynomial time algorithms
  for modules over finite dimensional algebras. In: Char, B.~W., Wang, P.~S.,
  K{\"u}chlin, W. (Eds.), ISSAC. ACM, pp. 68--74.

\bibitem[{Cohn and Umans(2013)}]{DBLP:conf/soda/CohnU13}
Cohn, H., Umans, C., 2013. Fast matrix multiplication using coherent
  configurations. In: Proceedings of the Twenty-Fourth Annual ACM-SIAM
  Symposium on Discrete Algorithms. SODA. SIAM, pp. 1074--1086.

\bibitem[{Comon et~al.(2008)Comon, Golub, Lim, and
  Mourrain}]{DBLP:journals/siammax/ComonGLM08}
Comon, P., Golub, G.~H., Lim, L.-H., Mourrain, B., 2008. Symmetric tensors and
  symmetric tensor rank. SIAM J. Matrix Analysis Applications 30~(3),
  1254--1279.

\bibitem[{de~Seguins~Pazzis(2010)}]{Pazzis2010}
de~Seguins~Pazzis, C., 2010. Invariance of simultaneous similarity and
  equivalence of matrices under extension of the ground field. Linear Algebra
  and its Applications 433~(3), 618 -- 624.

\bibitem[{DeMillo and Lipton(1978)}]{DL78}
DeMillo, R., Lipton, R., 1978. A probabilistic remark on algebraic program
  testing. Information Processing Letters 7~(4), 192--194.

\bibitem[{Edmonds(1967)}]{edm67}
Edmonds, J., 1967. Systems of distinct representatives and linear algebra.
  Journal of Research of the National Bureau of Standards 718~(4), 242 -- 245.

\bibitem[{Faug\`ere et~al.(2012)Faug\`ere, Lin, Perret, and Wang}]{DFPW11}
Faug\`ere, J.-C., Lin, D., Perret, L., Wang, T., 2012. {On enumeration of
  polynomial equivalence classes and their application to MPKC}. Finite Fields
  and Their Applications 18~(2), 283 -- 302.

\bibitem[{Faug{\`e}re and Perret(2006)}]{FaugerePerret06}
Faug{\`e}re, J.-C., Perret, L., 2006. {Polynomial Equivalence Problems:
  Algorithmic and Theoretical Aspects}. In: Vaudenay, S. (Ed.), EUROCRYPT. Vol.
  4004 of Lecture Notes in Computer Science. Springer, pp. 30--47.

\bibitem[{Gantmacher(1959)}]{grant}
Gantmacher, F., 1959. The Theory of Matrices, Vol. 1. Chelsea.

\bibitem[{Giraud(1972)}]{Giraud1972}
Giraud, J., 1972. \'{E}tude locale des singularit\'es. U.E.R. Math\'ematique,
  Universit\'e Paris XI, Orsay, cours de 3{\`e}me cycle, 1971--1972,
  Publications Math{\'e}matiques d'Orsay, No. 26.

\bibitem[{Green and Tao(2009)}]{DBLP:journals/cdm/GreenT09}
Green, B.~J., Tao, T., 2009. The distribution of polynomials over finite
  fields, with applications to the gowers norms. Contributions to Discrete
  Mathematics 4~(2).

\bibitem[{Grigorescu et~al.(2013)Grigorescu, Wimmer, and
  Xie}]{DBLP:journals/eccc/BhattacharyyaGRS11}
Grigorescu, E., Wimmer, K., Xie, N., 2013. Tight lower bounds for testing
  linear isomorphism. Electronic Colloquium on Computational Complexity (ECCC),
  17.

\bibitem[{Hironaka(1970)}]{Hironaka1970}
Hironaka, H., 1970. Additive groups associated with points of a projective
  space. Ann. of Math. (2) 92, 327--334.

\bibitem[{Kaltofen(1989)}]{Kaltofen89factorizationof}
Kaltofen, E., 1989. Factorization of polynomials given by straight-line
  programs. In: Randomness and Computation. JAI Press, pp. 375--412.

\bibitem[{Kayal(2011)}]{Kayal2011}
Kayal, N., 2011. Efficient algorithms for some special cases of the polynomial
  equivalence problem. In: Proceedings of the {T}wenty-{S}econd {A}nnual
  {ACM}-{SIAM} {S}ymposium on {D}iscrete {A}lgorithms. SIAM, Philadelphia, PA,
  pp. 1409--1421.

\bibitem[{Kayal(2012)}]{DBLP:conf/stoc/Kayal12}
Kayal, N., 2012. Affine projections of polynomials: extended abstract. In:
  Karloff, H.~J., Pitassi, T. (Eds.), STOC. ACM, pp. 643--662.

\bibitem[{Lang(2002)}]{Lang2002}
Lang, S., 2002. Algebra, 3rd Edition. Vol. 211 of Graduate Texts in
  Mathematics. Springer-Verlag, New York.

\bibitem[{Lidl and Niederreiter(1997)}]{finite_fields}
Lidl, R., Niederreiter, H., 1997. Finite fields, 2nd Edition. Vol.~20 of
  Encyclopedia of Mathematics and its Applications. Cambridge University Press,
  Cambridge, with a foreword by P. M. Cohn.

\bibitem[{Macario-Rat et~al.(2013)Macario-Rat, Pl\^ut, and
  Gilbert}]{MacariotRatPG2013}
Macario-Rat, G., Pl\^ut, J., Gilbert, H., 2013. New {I}nsight into the
  {I}somorphism of {P}olynomial {P}roblem {IP1S} and {I}ts {U}se in
  {C}ryptography. In: Sako, K., Sarkar, P. (Eds.), Advances in Cryptology -
  ASIACRYPT 2013. Vol. 8269 of Lecture Notes in Computer Science. Springer
  Berlin Heidelberg, pp. 117--133.

\bibitem[{Mackey et~al.(2005)Mackey, Mackey, and Tisseur}]{MaMaTi2005}
Mackey, D.~S., Mackey, N., Tisseur, F., 2005. Structured factorizations in
  scalar product spaces. SIAM J. Matrix Anal. Appl. 27~(3), 821--850.

\bibitem[{Matsumoto and Imai(1988)}]{C*}
Matsumoto, T., Imai, H., 1988. Public quadratic polynomial-tuples for efficient
  signature-verification and message-encryption. In: Advances in Cryptology --
  EUROCRYPT 1988. Vol. 330 of LNCS. Springer--Verlag, pp. 419--453.

\bibitem[{Matthews(1992)}]{Matthews1992}
Matthews, K.~R., 1992. A rational canonical form algorithm. Math. Bohemica 117,
  315--324.

\bibitem[{Mulmuley(2012)}]{DBLP:journals/cacm/Mulmuley12}
Mulmuley, K., 2012. The {GCT} program toward the {{\it P}} vs. {{\it NP}}
  problem. Commun. ACM 55~(6), 98--107.

\bibitem[{Mulmuley and Sohoni(2001)}]{DBLP:journals/siamcomp/MulmuleyS01}
Mulmuley, K., Sohoni, M.~A., 2001. Geometric {C}omplexity {T}heory {I}: {A}n
  {A}pproach to the {P} vs. {NP} and {R}elated {P}roblems. SIAM J. Comput.
  31~(2), 496--526.

\bibitem[{Newman(1967)}]{Newman1967}
Newman, M., 1967. Two classical theorems on commuting matrices. J. Res. Nat.
  Bur. Standards Sect. B 71B, 69--71.

\bibitem[{Patarin(1996)}]{DBLP:conf/eurocrypt/Patarin96}
Patarin, J., 1996. Hidden {F}ields {E}quations ({HFE}) and {I}somorphisms of
  {P}olynomials ({IP}): {T}wo {N}ew {F}amilies of {A}symmetric {A}lgorithms.
  In: Maurer, U.~M. (Ed.), EUROCRYPT. Vol. 1070 of Lecture Notes in Computer
  Science. Springer, pp. 33--48.

\bibitem[{Patarin et~al.(1998)Patarin, Goubin, and
  Courtois}]{DBLP:conf/eurocrypt/PatarinGC98}
Patarin, J., Goubin, L., Courtois, N., 1998. Improved algorithms for
  isomorphisms of polynomials. In: Nyberg, K. (Ed.), EUROCRYPT. Vol. 1403 of
  Lecture Notes in Computer Science. Springer, pp. 184--200.

\bibitem[{Perret(2004)}]{DBLP:journals/eccc/ECCC-TR04-116}
Perret, L., 2004. On the computational complexity of some equivalence problems
  of polynomial systems of equations over finite fields. Electronic Colloquium
  on Computational Complexity (ECCC) 116.

\bibitem[{Perret(2005)}]{DBLP:conf/eurocrypt/Perret05}
Perret, L., 2005. A fast cryptanalysis of the isomorphism of polynomials with
  one secret problem. In: Cramer, R. (Ed.), EUROCRYPT. Vol. 3494 of Lecture
  Notes in Computer Science. Springer, pp. 354--370.

\bibitem[{Saxena(2006)}]{Saxena2006}
Saxena, N., 2006. Morphisms of {R}ings and {A}pplications to {C}omplexity.
  Ph.D. thesis, Indian Institute of Technology Kanpur.

\bibitem[{Shanks(1973)}]{Shanks1973}
Shanks, D., 1973. Five number-theoretic algorithms. In: Proceedings of the
  {S}econd {M}anitoba {C}onference on {N}umerical {M}athematics ({U}niv.
  {M}anitoba, {W}innipeg, {M}an., 1972). Utilitas Math., Winnipeg, Man., pp.
  51--70. Congressus Numerantium, No. VII.

\bibitem[{Singla(2010)}]{Singla2010}
Singla, P., 2010. On representations of general linear groups over principal
  ideal local rings of length two. Journal of Algebra 324~(9), 2543--2563.

\bibitem[{Storjohann(1998)}]{Storjohann1998}
Storjohann, A., 1998. An {$O(n^3)$} algorithm for the {F}robenius normal form.
  In: Proceedings of the 1998 International Symposium on Symbolic and Algebraic
  Computation. ISSAC '98. ACM, New York, NY, USA, pp. 101--105.

\bibitem[{Strassen(1969)}]{Strassen1969}
Strassen, V., 1969. Gaussian elimination is not optimal. Numer. Math. 13,
  354--356.

\bibitem[{Tang and Xu(2012)}]{DBLP:conf/nss/TangX12}
Tang, S., Xu, L., 2012. Proxy signature scheme based on isomorphisms of
  polynomials. In: Xu, L., Bertino, E., Mu, Y. (Eds.), NSS. Vol. 7645 of
  Lecture Notes in Computer Science. Springer, pp. 113--125.

\bibitem[{Tang and Xu(2014)}]{Tang2014}
Tang, S., Xu, L., 2014. Towards provably secure proxy signature scheme based on
  isomorphisms of polynomials. Future Generation Computer Systems 30, 91 -- 97,
  special Issue on Extreme Scale Parallel Architectures and Systems,
  Cryptography in Cloud Computing and Recent Advances in Parallel and
  Distributed Systems, ICPADS 2012.

\bibitem[{Valiant(1979)}]{DBLP:journals/tcs/Valiant79}
Valiant, L.~G., 1979. The complexity of computing the permanent. Theor. Comput.
  Sci. 8, 189--201.

\bibitem[{von~zur Gathen and Gerhard(1999)}]{GaGe1999}
von~zur Gathen, J., Gerhard, J., 1999. Modern computer algebra. Cambridge
  University Press, New York.

\bibitem[{Wallenborn(2013)}]{Wallenborn2013}
Wallenborn, L.~A., 2013. Berechnung des {H}ilbert {S}ymbols, quadratische
  {Form-{\"A}quivalenz} und {F}aktorisierung ganzer {Z}ahlen. Master's thesis,
  Rheinische Friedrich-Wilhelms-Universit\"at.

\bibitem[{Wolf and Preneel(2011)}]{WP11}
Wolf, C., Preneel, B., 2011. Equivalent keys in multivariate quadratic public
  key systems. Journal of Mathematical Cryptology 4~(4), 375--415.

\bibitem[{Yang et~al.(2011)Yang, Tang, and Yang}]{DBLP:conf/ispec/YangTY11}
Yang, G., Tang, S., Yang, L., 2011. A novel group signature scheme based on
  mpkc. In: Bao, F., Weng, J. (Eds.), ISPEC. Vol. 6672 of Lecture Notes in
  Computer Science. Springer, pp. 181--195.

\bibitem[{Zippel(1979)}]{Z79}
Zippel, R., 1979. Probabilistic algorithms for sparse polynomials. In: Symbolic
  and algebraic computation (EUROSAM'79), Internat. Sympos. Vol.~72 of Lecture
  Notes in Computer Science. Springer Verlag, pp. 216--226.

\end{thebibliography}


%% \begin{thebibliography}{00}


%% \bibitem[ ()]{}

%% \end{thebibliography}
\appendix

\section{Proofs of \#\IPS}\label{ap:enum_proof}
In this appendix, we shall prove the dimension of the centralizer of a
matrix $J$, a Jordan normal form. This dimension, a consequence
of~\cite[Lemma~4.11]{Singla2010}, is used in
Section~\ref{s:enumeration} to determine an upper bound on the
counting problem of quadratic-\IPS. As stated by Singla, the proofs
only involve matrix multiplications are given in order for the paper to be
self-contained.

First, let us recall that the centralizer of $J$, a Jordan
block of size $s$ is the set of upper triangular Toeplitz matrices of
size $s\times s$. Indeed, if $X$ commutes with $J$, $XJ-JX$ is as such
\[X\,J-J\,X=\pare{\begin{smallmatrix}
    -x_{2,1} &x_{1,1}-x_{2,2} &\ldots &x_{1,n-1}-x_{2,n}\\
    \vdots &\vdots &&\vdots\\
    -x_{n,1} &x_{n-1,1}-x_{n,2}&\cdots&x_{n-1,n-1}-x_{n,n}\\
    0 &x_{n,1} &\ldots &x_{n,n-1}
  \end{smallmatrix}}=0.\]

This small result is used in the following Lemma to determine the
centralizer of a Jordan normal form.

\begin{lemma}\label{lm:cent}
  Let $J$ be a Jordan normal form. For $1\leq i\leq r$, let us denote
  $J_i$ the $i$th block of $J$ and let us assume it is associated with
  eigenvalue $\zeta_i$ and it is of
  size $s_i$. Let $X=(X_{i,j})_{1\leq i,j\leq r}$ be a block-matrix, with
  $X_{i,j}\in {\LL\pare{\zeta_1,\ldots,\zeta_r}}^{s_i\times s_j}$,
  that commutes with $J$. If $\zeta_i=\zeta_j$, then $X_{i,j}$ is
  an upper triangular Toeplitz matrix whose nonnecessary zero
  coefficients are the one on the first $\min(s_i,s_j)$ diagonals.
  Otherwise, $X_{i,j}=0$.
\end{lemma}
\begin{proof}
  We assume that $r=2$. If
  $X\,J-J\,X=\pare{\begin{smallmatrix}X_{1,1}\,J_1-J_1\,X_{1,1}
      &X_{1,2}\,J_2-J_1\,X_{1,2}\\
      X_{2,1}\,J_1-J_2\,X_{2,1}
      &X_{2,2}\,J_1-J_1\,X_{2,2}\end{smallmatrix}}=0$, then $X_{1,1}$
  commutes with $J_1=J_{\zeta_1,s_1}$ and $X_{2,2}$ with
  $J_2=J_{\zeta_2,s_2}$. Thus they are upper triangular Toeplitz
  matrices.
  
  From $X_{2,1}\,J_2-J_1\,X_{2,2}$, one deduces that
  $(\zeta_1-\zeta_2)\,x_{s_1+s_2,1}=0$, hence either
  $\zeta_1=\zeta_2$ or $x_{s_1+s_2,1}=0$.  If
  $\zeta_1\neq\zeta_2$, then step by step, one has
  $X_{1,2}=0$. Assuming $\zeta_1=\zeta_2$, then step by step, one
  has $x_{s_1+i,1}=0$ for $i>1$ and since
  $x_{s_1+i+1,j+1}-x_{s_1+i,j}=0$ for all $i,j$, one has in fact that
  $X_{1,2}$ is a upper triangular Toeplitz matrix with potential
  nonzero coefficients on the first $\min(s_1,s_2)$ diagonals. The
  same argument applies to $X_{2,1}$.

  The case $r>2$ is an easy generalization of this result.
\end{proof}
From this lemma, we can deduce easily the dimension of the centralizer
of a matrix.
\begin{lemma}\label{lm:dim_cent}
  Let $H\in\K^{n\times n}$ be a matrix and let $J$ be its normal
  Jordan form. Assuming the blocks of $J$ associated to $\zeta_i$ are
  $J_{\zeta_i,s_{i,1}},\ldots,J_{\zeta_1,s_{i,d_1}}$ with
  $s_{i,1}\leq\cdots\leq s_{i,d_i}$ for $i,1\leq i\leq r$, then the
  centralizer of $H$ is a
  $\K$-vector subspace of $\K^{n\times n}$ of dimension no more than
  $\sum_{1\leq i\leq r}\sum_{1\leq j\leq d_i}(2\,d_i-2\,j+1)s_{i,j}$.
\end{lemma}
\begin{proof}
  Let $\LL$ be the smallest field over which $J$ is defined. It is
  clear that the centralizer of $H$ over $\LL$, denoted $\mcal{W}$,
  contains $\Cent(H)\otimes\LL$. Hence,
  $\dim_{\K}\Cent(H)=\dim_{\LL}(\Cent(H)\otimes\LL)
  \leq\dim_{\LL}\mcal{W}$.

  Now, let $X=(X_{i,j})_{1\leq i,j\leq
    d_1+\cdots+d_r}\in\mcal{V}$. From Lemma~\ref{lm:cent}, there are
  $\sum_{1\leq i\leq r}\sum_{1\leq j\leq d_i} s_{i,j}$ free
  parameters for the diagonal blocks of $X$ and $2\,\sum_{1\leq i\leq
    r}\sum_{1\leq j<k\leq d_i}\min(s_{i,j}, s_{i,k})
  =2\sum_{1\leq i\leq r}\sum_{1\leq j\leq d_i}(d_i-j)s_{i,j}$ free
  parameters for the off-diagonal blocks of $X$. This concludes the
  proof.
\end{proof}
As a consequence, the number of invertible matrices in $\Cent(H)$ is
bounded from above by
\[q^{\pare{\sum_{1\leq i\leq r}\sum_{1\leq j\leq d_i}(2\,d_i-2\,j+1)\,s_{i,j}}}-1,\]
as stated in Corollary~\ref{cor:enum}.

\end{document}